\documentclass[11pt]{article}

\usepackage{hyperref}
\hypersetup{
    colorlinks=true,
    linkcolor=blue,
    filecolor=magenta,   
    urlcolor=blue,
}

\usepackage{latexsym}
\usepackage{amsmath,amssymb,amsthm}

\usepackage{setspace}
 \newcommand{\beqn}{\begin{eqnarray}}
 \newcommand{\eeqn}{\end{eqnarray}}
 \newcommand{\be}{\begin{equation}}
 \newcommand{\ee}{\end{equation}}
 \newcommand{\ba}{\begin{array}}
 \newcommand{\ea}{\end{array}}

 \newcommand{\pa}{\partial}

 \newcommand{\ds}{\displaystyle}
 \newcommand{\la}{\label}
 \newcommand{\rIm}{{\rm Im\5}}
 \newcommand{\rRe}{{\rm Re\5}}

\newcommand{\ov}{\overline}

\newcommand{\ti}{\tilde}
\newcommand{\cE}{{\cal E}}
\newcommand{\cH}{{\cal H}}
\newcommand{\bJ}{{\bf J}}
\newcommand{\ve}{\varepsilon}
\newcommand{\De}{\Delta}

\newcommand{\al}{\alpha}
\newcommand{\ga}{\gamma}

\newcommand{\si}{\sigma}
\newcommand{\om}{\omega}
\newcommand{\Om}{\Omega}
\newcommand{\na}{\nabla}
\newcommand{\lam}{\lambda}
\newcommand{\Lam}{\Lambda}
\newcommand{\5}{{\hspace{0.5mm}}}
\newcommand{\R}{\mathbb{R}}

\newcommand{\C}{\mathbb{C}}
\newcommand{\N}{\mathbb{N}}

\newtheorem{theorem}{Theorem}[section]

\newtheorem{defin}[theorem]{Definition}
\newtheorem{lemma}[theorem]{Lemma}
\newtheorem{remark}[theorem]{Remark}
\newtheorem{cor}[theorem]{Corollary}
\newtheorem{pro}[theorem]{Proposition}

\begin{document}
\begin{center}
{\Large On soliton asymptotics for  Maxwell--Lorentz 
 \\
 equations with rotating particle}
 \medskip 

E.A. Kopylova
 \medskip
 \\
{\it
Institute of Mathematics of BOKU University, Vienna, Austria\\
}
 elena.kopylova@boku.ac.at
\begin{abstract}
We consider the 2D Maxwell--Lorentz equations with an extended charged rotating particle. 
The system admits soliton solutions corresponding to a particle moving with a constant velocity and rotating with a constant angular velocity. 
Our main result is the asymptotic stability of  the solitons.
\end{abstract}

\end{center}
\setcounter{equation}{0}
\section{Introduction} \label{Int}
The Maxwell--Lorentz equations describe the motion of an extended charged particle with charge density $\rho(x)$
in the Maxwell field $(E(x, t), B(x, t))$.  We  express the Maxwell field in 
potentials $A(x,t)$,  and $\Phi(x,t)$:
\[
B(x,t)=\na\cdot JA(x,t),\qquad E(x,t)=-\dot A(x,t)-\na \Phi(x,t),\quad J=\begin{pmatrix} 0&1\\-1& 0\end{pmatrix},\quad \na\cdot A(x,t)=0.
\]
Then  the  2D Maxwell--Lorentz equations  with  a rotating particle read as follows \cite[(1.6), (3.2)]{KK2024}:
\be\la{mls3}
\left\{\ba{llll}
\dot A(x,t)=\Pi(x,t),\qquad
\dot \Pi(x,t)=\Delta A(x,t)-\dot\varphi J\varrho(x-q)+{\cal P}[\dot q\rho(x-q)],\\
m\dot q(t)=p-\langle A(x,t),\rho(x-q(t))\rangle,
\\
\dot p(t)=-\langle  A(x,t)\cdot \dot q(t),\nabla\rho(x-q(t))\rangle+\dot\varphi \langle \na\cdot JA(x,t),  \varrho(x-q(t))\rangle, 
\\
I\dot\varphi(t)=M(t)+\langle A(x,t),J\varrho(x-q(t))\rangle,\qquad  \dot M(t)=0.
\ea\right|,~x\in\R^2.
\ee
Here $q$ is the  location of the particle, $\varphi$  is the  angular coordinate,  $m >0$ is the mass of the particle,  $I>0 $ is  the moment of inertia,  
$\varrho(x)=x\rho(x)$,
and   the brackets $\langle,\rangle$ denote the inner product in the  Hilbert spaces $L^2(\R^2)$ or $L^2(\R^2)\otimes\R^2$. 
Further, $p=m\dot q+\langle A(x),\rho(x-q)\rangle$ is the  particle momentum,  $M=I\dot\varphi+\langle A(x),J\varrho(x-q)\rangle$ is the angular momentum, 
and ${\cal P}$  is the projection onto the space of divergence-free vector fields,
which in Fourier space is given by
\be\la{Pi-e}
\widehat{{\cal P}a}(k)=\hat a(k)-\frac{\hat a(k)\cdot k}{k^2}k=\frac{\hat a(k)\cdot Jk}{k^2}Jk.
\ee
We assume that the real-valued  charge density $\rho(x)$ satisfies the following conditions: 
\be\la{rosym}
\rho\in C_0^\infty(\R^2),\quad\rho(x)=\rho_{rad}(|x|),\quad\rho(x)\not\equiv 0,\quad \int_{\R^2} x^{\al}\rho(x) dx=0,\qquad  |\al|\le 2. 
\ee
These conditions imply that
\be\la{zero2}
\hat\rho(k)={\cal O}(|k|^4),\qquad k\to 0.
\ee
The system  admits  soliton solutions 
\be\la{solY}
Y_{a,v,\theta,\omega}(t)=(A_{v,\omega}(x-vt-a),\Pi_{v,\omega}(x-vt-a), vt+a,  p_{v,\omega}, {\omega} t+\theta, M_{v,\omega})
\ee
corresponding to a particle that  moves with constant  velocity $v\in\R^2$, $0\le |v|<1$ 
and rotates with constant angular velocity $\omega\in\R$.  

Our main result is as follows:
for finite energy solutions to \eqref{mls3} with  initial data  sufficiently close to  a soliton  with   velocity $v$ and angular velocity $\omega$,
the following  asymptotics hold:
\beqn\nonumber
&&\dot q(t)\to v_\pm\approx v,\quad   \dot\varphi(t)\to \omega_\pm\approx \omega,\\ 
\la{main-as}
&&\left(\ba{cc} A(x,t)\\\Pi(x,t)\ea\right)=\left(
\ba{cc} A_{v_{\pm},{\omega}_{\pm}}(x-v_{\pm}t-a_{\pm})\\ \Pi_{v_{\pm},{\omega}_{\pm}}(x-v_{\pm}t-a_{\pm})\ea\right)
+W_0(t)\Psi_{\pm}+r_{\pm}(x,t),\quad t\to\pm\infty.
\eeqn
Here $W_0(t)$ is the dynamical group of the free wave equation, 
$\Psi_{\pm}$ are the corresponding asymptotic scattering states, 
and the remainder $r_{\pm}(x,t)$ converges to zero in the global energy norm.

We prove these asymptotics  under a suitable  spectral condition \eqref{M-condition}.  
In particular, this condition holds for solitons with $v\ne 0$,  small $\omega$, and large $I$, and for solitons with
$v=0$ and small $\om$ under the  Wiener condition $\hat\rho(k)\ne 0$ for $k\ne 0$; see Appendix  B.

 \smallskip
Let us  now comment on previous results on the asymptotic stability of solitons  for wave-particle systems.
Soliton scattering asymptotics  of type \eqref{main-as} were first obtained  in \cite{IKV2006} for a moving relativistic  particle coupled to the 3D Klein--Gordon field. 
In \cite{KK2006},\cite {KKS2011}, \cite{IKV2012}, \cite{KKS2018} this result was extended to moving relativistic  and nonrelativistic particles 
coupled to the 3D Dirac, Schr\"odinger and wave equations. 
The asymptotic stability of moving solitons (with $\omega\equiv 0$)   was established in  \cite{IKS2011} for the  3D 
 Maxwell--Lorentz system. The techniques of \cite{IKS2011}  rely essentially on the strong Huygens' principle, 
and the corresponding resolvent is smooth at the threshold $\lam = 0$.   
 All these results were obtained under a Wiener-type condition on the charge density: $\hat \rho(k)\ne 0$ for $k\ne 0$.
For related surveys, see \cite{CKK2023, KK2020, KK2022}.

 In  \cite{KK2023},  the 2D Maxwell--Lorentz system  \eqref{mls3} with  a moving nonrelativistic  particle without rotation 
was considered (i.e., the first three equations of  \eqref{mls3} with $\dot\varphi=0$).
In this case, the absence of the strong Huygens' principle   required a new  approach based 
 on the Puiseux expansion of the corresponding  resolvent operator. 
Condition  \eqref{rosym} is imposed to provide  sufficient spatial decay of the solitons.  
Note that  in  \cite{KK2023} the  soliton scattering  asymptotics of type \eqref{main-as} were obtained  without the Wiener condition.

In the present paper, we extend the   result of \cite{KK2023}  to the system \eqref{mls3} with  a moving and rotating particle.  
In the case $\omega\ne 0$,  the functions $A_{v,\omega}$ and  $\Pi_{v,\omega}$  have   weaker decay.
Despite this, we   strengthen the result of \cite{KK2023} by reducing the exponent $\beta>7/2$ 
in our main Theorem \ref{main}  to  $\beta> 5/2$.  This becomes possible due to a suitable extension of  the domain of the corresponding resolvent.

We develop the  strategy of \cite{BKKS, BS, IKV2006} based on  symplectic geometry methods for  Hamiltonian systems in Hilbert spaces 
and on the spectral theory of nonselfadjoint operators.
The  key point in the proof  of asymptotics \eqref{main-as} is a dispersive decay for the  linearized dynamics $e^{L_{v,\om}t}$ on a soliton
(Theorem \ref{lindecay}). The decay holds in the direction   symplectically orthogonal  to the solitary manifold.

The orbital stability for the system \eqref{mls3} with  a rotating particle was obtained in \cite{KK2024}. 
The global attraction to the solitary manifold was proved in \cite{KK2023+}
 for the system \eqref{mls3}  with a rotating  particle  located at the origin, i.e.,  with $q(t)\equiv 0$. 

The paper is organized as follows.
In Section \ref{M-Res} we prove asymptotics \eqref{main-as} for solutions to the nonlinear equations \eqref{mls3}, relying on the dispersive decay of the linearized dynamics. This decay is established separately, in Section \ref{lin-dyn}.
\setcounter{equation}{0}
\section{Asymptotic behavior of solutions to the nonlinear equation \eqref{mls3}}\label{M-Res}
\subsection{Preliminary facts and formulation of the result}
First, we introduce a phase space for the system \eqref{mls3}.  Let  $L^2_\beta(\R^2)$, $\beta\in\R$, be  the  weighted Sobolev space with the norm
$\Vert f\Vert_{L^2_{\beta}(\R^2)}=\Vert (1+|x|)^{\beta} f(x)\Vert_{L^2(\R^2)}$, and let 
${\bf L}^2_{\beta}$ be  the space of divergence-free vector  fields  in  $L^2_{\beta}(\R^2)\otimes\R^2$.
Denote by ${\cal F}_{\beta}$ the space  of  divergence-free vector  fields $(A,\Pi)$ and by ${\cal E}_{\beta}$  the space 
of  states $Y=(A,\Pi,q,p,\varphi, M)$ with the norms
\[
\Vert (A,\Pi)\Vert_{{\cal F}_\beta}=\Vert \na A\Vert_{{\bf L}^2_{\beta}}+\Vert \Pi\Vert_{{\bf L}^2_{\beta}},\qquad
\Vert Y\Vert_{\beta}=\Vert \na A\Vert_{{\bf L}^2_{\beta}}+\Vert \Pi\Vert_{{\bf L}^2_{\beta}}+|q|+|p|+|\varphi|+|M|.
\]
Denote ${\cal F}={\cal F}_0$, ${\cal E}={\cal E}_0$.
The system (\ref{mls3}) can be written  as the Hamiltonian system
\be\la{canH+}
\dot Y={\bf J}{\cal D}\cH(Y),~~ Y=(A,\Pi, q, p,\varphi, M)\in {\cal E},~~ 
{\bf J}=\begin{pmatrix}
0 & I_2 & 0 & 0 & 0  \\
-I_2 & 0 & 0 & 0& 0 \\
0 & 0 & 0 & I_2& 0  \\
0 & 0 & -I_2 & 0 & 0 \\
0 & 0 & 0 & 0& J  
\end{pmatrix},
 ~~ I_2=\begin{pmatrix} 1& 0\\ 0&1\end{pmatrix},
\ee
where ${\cal D}{\cal H}$ is the Fr\'echet derivative of the Hamiltonian functional
$$
\cH(Y)=\ds\frac 12\int[|\Pi(x)|^2+|\na A(x)|^2]dx+\frac 1{2}m\dot q^2(t)+\frac 1{2} I\dot\varphi^2(t).
$$
Here $\dot q=\frac{1}{m}\big(p-\langle A,\rho(x-q)\rangle\big)$, $\dot\varphi=\frac{1}{I}\big(M+\langle A,J\varrho(x-q)\rangle\big)$ 
 in agreement with  (\ref{mls3}).
\begin{pro} 
Let \eqref{rosym} hold, and let $Y_0=(A_0, \Pi_0, q_0, p_0,\varphi_0, M_0)\in {\cal E}$. Then\\
(i) there exists a unique solution $Y(t)\in C(\R, {\cal E})$ to the Cauchy problem for \eqref{mls3};\\
(ii) the energy is conserved: $\cH(Y(t))=\cH(Y_0)$ for  $t\in \R$.
\end{pro}
The proof is similar to that of \cite[Lemma A.2]{KS2000}.
Solitons of the system (\ref{mls3}) are solutions of the form \eqref{solY} with 
\be\la{solvom}
p_{v,\omega}=mv+\langle A_{v,\omega},\rho\rangle,\quad M_{v,\omega} =I{\omega}  -\langle A_{v,\omega},J\varrho\rangle.
\ee
Substituting \eqref{solY} into the first two equations of \eqref{mls3}, we obtain
\be\la{Hs5}
\Pi_{v,\omega}(y)=-(v\cdot\na)A_{v,\omega}(y), \qquad 
-(v\cdot\na)\Pi_{v,\omega}(y)=\De A_{v,\omega}(y)-{\omega} J\varrho(y)+{\cal P}[v\rho(y)].
\ee
Hence,
$\Delta  A_{v,\omega}(y)-(v\cdot\na)^2 A_{v,\omega}(y)={\omega} J\varrho(y)-{\cal P}[v\rho(y)]$.
For $|v|<1$, in the  Fourier representation, we get 
\be\la{solit3}
\hat A_{v,\omega}(k)=
\frac{(v-\frac{(v\cdot k)k}{k^2})\hat \rho(k)}{\hat D_0}-\frac{{\omega} J\hat\varrho(k)}{\hat D_0},~~{\rm where}~~~\hat D_0:=k^2-(v\cdot k)^2.
\ee
By \eqref{rosym}, $\na A_{v,\omega}(y),\Pi_{v,\omega}(y)\sim |y|^{-4}$ as $|y|\to \infty$ (cf.  \cite[Lemma 3.1]{KK2023}).
Hence, 
$\na A_{v,\omega},\Pi_{v,\omega}\in {\bf L}^2_{\beta}$  with $\beta<3$.
Denote   ${\cal V}=\{v\in\R^2:|v|<1\}$.
\begin{defin}
$S(\sigma):=(A_{v,\omega}(x-b),\Pi_{v,\omega}(x-b),b, p_{v,\om},\ga, M_{v,\omega})$, where
$\sigma:=(b,v,\ga, \omega)$ with $b\in\R^2$, $v\in {\cal V}$, and $\ga,{\omega}\in\R$.
\end{defin}
The soliton solution \eqref{solY} can be represented as  $S(\si(t))$ with
$\si(t)=(vt+a,v,{\omega} t+\theta, {\omega})$.
\begin{defin}
The solitary manifold is the set ${\cal S}:=\{S(\si):\si\in\Sigma:=\R^2\times {\cal V}\times\R\times\R\}$.
\end{defin}
Our main result is the following theorem:
\begin{theorem}\la{main}
Let condition \eqref{rosym} and the spectral condition \eqref{M-condition} hold, and let $5/2<\beta<3$.
Suppose that $Y_0\in {\cal E}_{\beta}$ is sufficiently close to the solitary manifold in the following sense:
\be\la{close}
Y_0=S(\si_0)+Z_0,\qquad d_\beta=\Vert Z_0\Vert_{\beta}\ll 1.
\ee
Let $Y(t)\in C(\R, {\cal E})$ be the solution to \eqref{mls3} with initial data $Y_0$. Then:  
\beqn\la{qq-as}
&&\dot q(t)= v_{\pm}+{\cal O}(|t|^{-3/2}),\quad  q(t)=v_{\pm}t+a_{\pm}+{\cal O}(|t|^{-1}),\quad t\to\pm\infty,\\
\la{wphi-as}
&&\dot\varphi(t)={\omega}_{\pm}+{\cal O}(|t|^{-3/2}),\quad \varphi(t)={\omega}_{\pm}t+\theta_{\pm}+ {\cal O}(|t|^{-1}),\quad t\to\pm\infty,\\
\la{AP-as}
&&\left(\ba{cc} A(x,t)\\\Pi(x,t)\ea\right)=\left(
\ba{cc} A_{v_{\pm},{\omega}_{\pm}}(x-v_{\pm}t-a_{\pm})\\ \Pi_{v_{\pm},{\omega}_{\pm}}(x-v_{\pm}t-a_{\pm})\ea\right)
+W_0(t)\Psi_{\pm}+r_{\pm}(x,t),\quad t\to\pm\infty,
\eeqn
and the remainder $r_{\pm}(x,t)$ converges to zero in the global energy norm:
\be\la{r-as}
\Vert r_{\pm}(t)\Vert_{\cal F}={\cal O}(|t|^{-1/2}),\quad t\to\pm\infty.
\ee
\end{theorem}
\subsection{Symplectic structure}\label{SymS}
Here we introduce  a symplectic structure on the phase space ${\cal E}$.
We define  a symplectic form $\Om$ on ${\cal E}$  as follows:
\be\la{OmJ}
\Om(Y_1,Y_2)=\langle Y_1,\bJ Y_2\rangle,\,\,\,Y_1,Y_2\in {\cal E},
\ee
where
$\langle Y_1,Y_2\rangle:=\langle A_1, A_2\rangle+\langle\Pi_1,\Pi_2\rangle+q_1 \cdot q_2+p_1\cdot  p_2+ \varphi_1 \varphi_2+M_1  M_2$.
\begin{defin}
i) We write $Y_1\nmid Y_2$  if   $Y_1$ is symplectically orthogonal to $Y_2$, i.e.,  if $\Om(Y_1,Y_2)=0$.
\\
ii) A projection operator ${\bf P}:{\cal E}\to{\cal E}$ is  symplectically orthogonal if $Y_1\nmid Y_2$ for all $Y_1\in\mbox{\rm Ker}\5{\bf P}$ and
all $Y_2\in\mbox {\rm Im}\5{\bf P}$.
\end{defin}
Denote by  ${\cal T}_{\si}{\cal S}$  the tangent space to the manifold ${\cal S}$ at the point $S(\si)$.
The vectors 
\be\la{inb}
\left.\ba{rclrclrclrclrclrcl}
\tau_j(v,\om):=\pa_{b_j}S(\si)=&(-\pa_j A_{v,\om}(y),&-\pa_j\Pi_{v,\om}(y),&e_j,& 0,& 0,& 0)
\\
\tau_{j+2}(v,\om):=\pa_{v_j}S(\si)=&(\pa_{v_j}A_{v,\om}(y),&\pa_{v_j}\Pi_{v,\om}(y),& 0,& \pa_{v_j}p_{v,\om}, &  0, & \pa_{v_j}M_{v,\om})
\\
\tau_{5}(v,\om):=\pa_{\gamma}S(\si)=&( 0,& 0,& 0,& 0,& 1,& 0)
\\
\tau_{6}(v,\om):=\pa_{\om}S(\si)=&(\pa_{\om}A_{v,\om}(y),& \pa_{\om}\Pi_{v,\om}(y),& 0,& 0,& 0,&\pa_{\om}M_{v,\om})
\ea\right.
\ee
form a basis in ${\cal T}_{\si}{\cal S}$. Here  $j=1,2$,  $e_1=(1,0)$,  $e_2=(0,1)$.
\begin{lemma}\la{Ome}
The matrix ${\bf\Om}={\bf\Om}(v,\om)$ with entries  $\Om(\tau_j(v, \om),\tau_l(v,\om))$ is non-degenerate for any $v\in {\cal V}$ and ${\om}\in\R$.
\end{lemma}
\begin{proof}
Let us compute the entries $\Om_{j,l}:=\Om(\tau_j(v,\om),\tau_l(v,\om))$  of the  matrix ${\bf\Om}(v,\om)$.
We write  $\hat A=\hat A_{v,\om}$, $\hat \Pi=\hat\Pi_{v,\om}$, $p=p_{v,\om}$, $M=M_{v,\om}$, $(vk)=v\cdot k$.
By \eqref{solit3},  $\hat A(k)=\Gamma(k)+i\Upsilon(k)$, where 
\be\la{APi}
\Gamma(k)=\Big[\frac1{\hat D_0}v -\frac{(vk)}{k^2\hat D_0}k\Big]\hat \rho(k),\qquad 
\Upsilon(k)=\frac{\om}{\hat D_0}J\na\hat\rho(k).
\ee
Note that  $\Gamma$ is even, and $\Upsilon$ is  odd.  For $ j,l=1,2$, formulas \eqref{OmJ}--\eqref{inb} imply 
\beqn\nonumber
\Om_{jl}&\!\!\!=\!\!\!&-2i\!\int\! k_jk_l(vk)\big[\Gamma+i\Upsilon\big]\cdot\big[\Gamma-i\Upsilon\big]dk
=-2i\!\int \!k_jk_l(vk)\big[|\Gamma|^2+|\Upsilon|^2\big]dk=0,\\
\nonumber
\Om_{j+2,j+2}
&\!\!\!=\!\!\!&-2i\!\int\!\Big[(vk)\big[|\pa_{v_j}\Gamma|^2+|\pa_{v_j}\Upsilon|^2\big]
+k_j\big[\Gamma\!\cdot \pa_{v_j}\Gamma\!+\Upsilon\!\cdot \pa_{v_j}\Upsilon\big]\Big] dk=0.
\eeqn
Obviously, $\Om_{15}=\Om_{25}=0$. Further, \eqref{solvom} and  \eqref{APi}  imply
\beqn\la{paA2}
\pa_{v_j}\Gamma&=&\frac{1}{\hat D_0}\Big[e_j+\frac{2(vk)k_j}{\hat D_0}v-\frac{k_j(k^2+(vk)^2)}{k^2\hat D_0}k\Big]\hat\rho,\qquad
\pa_{v_j}\Upsilon=\frac{2{\omega}(vk)k_j J\na\hat\rho}{\hat D_0^2},\\
\la{paP}
\pa_{v_j}p&=& e_j\big[m+\int\frac{|\hat\rho|^2dk}{\hat D_0}\big]+\int\frac{2|\hat\rho|^2(vk)k_j}{\hat D_0^2}vdk
-\int\frac{|\hat\rho|^2(k^2+(vk)^2)k_j}{k^2\hat D_0^2}kdk,\\
\la{Mvw}
\pa_{v_j}M&=&\om\int\frac{2k_j(vk)|\na\hat\rho|^2}{\hat D_0^2} dk,\qquad
\pa_{\om}M=I+\int\frac{|\na\hat\rho|^2}{\hat D_0} dk.
\eeqn
By rotational symmetry, we may choose $v=(|v|,0)$. In this case,  \eqref{APi} and  \eqref{paA2} give
\be\la{solit32+} 
\Gamma=\frac{|v|k_2}{k^2\hat D_0}\hat\rho Jk,\quad \Upsilon=\frac{\om}{\hat D_0} J\na\hat\rho, \quad
|\Gamma|^2=\frac{v^2k_2^2}{k^2\hat D_0^2}|\hat\rho|^2,\quad |\Upsilon|^2=\frac{\om^2}{\hat D_0^2}|\na\hat\rho|^2,
\ee
\vspace{-8mm}
\beqn\la{A1A1}
\pa_{v_1}\Gamma&\!\!\!=\!\!\!&\frac{k_2(k^2+v^2k_1^2)}{k^2\hat D_0^2}\hat\rho Jk,\quad 
\pa_{v_2}\Gamma=\frac{k_1(2v^2k^2-v^2k_1^2-k^2)}{k^2\hat D_0^2}\hat\rho Jk,\\
\la{A1A1++}
\pa_{v_1}\Upsilon&\!\!\!=\!\!\!&
\frac{2{\omega}|v|k_1^2}{\hat D_0^2}J\na\hat\rho,\qquad  \pa_{v_2}\Upsilon=\frac{2{\omega} |v|k_1k_2}{\hat D_0^2}J\na\hat\rho,\qquad
 \pa_{\omega}\Upsilon=\frac{1}{\hat D_0}J\na\hat\rho.
\eeqn
Now formulas \eqref{inb}, \eqref{paP} and \eqref{solit32+}--\eqref{A1A1++} imply  that
$\Om_{14}=\Om_{23}=\Om_{26}=\Om_{36}=\Om_{45}=0$.
Moreover,  $\Om_{jl}=-\Om_{lj}$ for $j\ne l$. Therefore,
\be\la{Om-fin}
{\rm det}\,{\bf \Om}(v,\omega)={\rm det}\begin{bmatrix}
0&0& \Om_{13} &0 &0&\Om_{16}\\
0&0& 0 & \Om_{24} &0&0\\
\Om_{31} & 0 &0&\Om_{34}& \Om_{35}& 0\\
0 & \Om_{42} &\Om_{43}& 0& 0& \Om_{46}\\
0&0& \Om_{53} & 0 &0&\Om_{56}\\
\Om_{61}&0&0 &\Om_{64}  &\Om_{65}&0
\end{bmatrix}
=\Om_{24}^2\,\bigl(\Om_{13}\Om_{56}+\Om_{16}\Om_{35}\bigr)^2.
\ee
By \eqref{paP} and \eqref{solit32+}--\eqref{A1A1++},
\beqn\nonumber
 \Om_{24}&=&\langle k_2\hat A, |v|k_1\pa_{v_2}\hat A+k_2\hat A \rangle+|v|\langle k_1k_2\hat A, \pa_{v_2}\hat A\rangle+e_2\cdot \pa_{v_2}p\\
\la{24} 
  &=&m+\int X_{24}|\hat\rho|^2dk+{\omega}^2\int \Big(\frac{k_2^2}{\hat D_0^2}+\frac{4v^2k_2^2k_1^2}{\hat D_0^3}\Big)|\na\hat\rho|^2dk>0
 \eeqn
 since
 \beqn\nonumber
X_{24}&\!\!\!:=\!\!\!&\frac{v^2k_2^2(k^2-k_1^2)}{k^2\hat D_0^2}+\frac{2v^2k_1^2k_2^2(2v^2k^2-k^2-v^2k_1^2)}{k^2\hat D_0^3}
-\frac{k_2^2}{\hat D_0^2}-\frac{v^2k_1^2k_2^2}{k^2\hat D_0^2}+\frac{1}{\hat D_0}\\
\nonumber
&\!\!\!=\!\!\!&\frac{k_2^2}{\hat D_0^3}\Big[v^2(k^2-k_1^2)+2v^4k_1^2+(v^2k_1-k_1)^2-k_1^2\Big]+\frac{1}{\hat D_0^3}\Big[k^2k_1^2+v^4k_1^4-2v^2k_1^2k^2\Big]\\
\nonumber
&\!\!\!\ge\!\!\!&\frac{1}{\hat D_0^3}\Big[v^4k_2^4+2v^4k_1^2k_2^2+v^4k_1^4+k_1^4-2v^2k_1^2k^2\Big]
=\frac{(v^2k^2-k_1^2)^2}{\hat D_0^3}\ge 0.
\eeqn
Further, \eqref{paP}, \eqref{solit32+}--\eqref{A1A1++}, and  \eqref{Mvw} imply that  
$\Om_{56}=\pa_{\omega} M=I+\int\frac{|\na\hat\rho|^2}{\hat D_0} dk>0$,
\beqn\nonumber
\Om_{13}&\!\!=\!\!&\langle k_1\hat A, |v|k_1\pa_{v_1}\hat A+k_1\hat A \rangle+|v|\langle k_1^2\hat A, \pa_{v_1}\hat A\rangle+e_1\cdot \pa_{v_1}p
=m+4v^2\om^2\!\int\!\frac{k_1^4|\na\hat\rho|^2}{\hat D_0^3}dk\\
\nonumber
&\!\!+\!\!&\om^2\!\int\!\frac{k_1^2|\na\hat\rho|^2}{\hat D_0^2}dk
+\int\!\!\Big(\frac{v^2k_1^2k_2^2}{k^2\hat D_0^2}+\frac{2v^2k_1^2k_2^2(k^2\!+\!v^2k_1^2)}{k^2\hat D_0^3}
+\frac{k^2k_2^2\!+\!v^2k_1^2k_2^2}{k^2\hat D_0^2}\Big) |\hat\rho|^2dk>0,\\
\nonumber
\Om_{35}&\!\!=\!\!&-\pa_{v_1}M=-2|v|\om\!\int\!\frac{k_1^2|\na\hat\rho|^2}{\hat D_0^2} dk,\qquad
 {\Om}_{16}=2\langle k_1\hat A,|v|k_1\pa_{\om}\hat A\rangle =2|v|\om\!\int\!\frac{k_1^2|\na\hat\rho|^2}{\hat D_0^2} dk.
\eeqn
Note that
$$
|\Om_{35} {\Om}_{16}|=4v^2\om^2\big(\int\frac{k_1^2|\na\hat\rho|^2}{\hat D_0^2} dk\big)^2
\le 4v^2\om^2 \!\int\!\frac{k_1^4|\na\hat\rho|^2}{\hat D_0^3}dk\cdot \int\frac{|\na\hat\rho|^2}{\hat D_0} dk< \Om_{13}\Om_{56}.
$$
Hence ${\rm det}\,{\bf \Om}(v,\omega)>0$ by \eqref{Om-fin} and \eqref{24}.
\end{proof}
Now we show that  a ``symplectically orthogonal projection''
onto ${\cal S}$ is well defined in a small neighborhood of  ${\cal S}$. 
Denote $v(Y)=\frac 1m(p-\langle A(x),\rho(x-q)\rangle)$ and  ${\omega} (Y)=\frac 1I(M+\langle A,J\varrho(x-q)\rangle)$. 
\begin{defin}
 For any $\al\in\R$, $0<\ov v<1$, and $\ov\om>0$,  we denote  
$\cE_\al(\ov v,\ov\om)=\{Y\in\cE_\al:|v(Y)|\le\ov v, ~|\om(Y)|\le\ov\om\}$ 
and  $\Sigma(\ov v,\ov\om)=\{\si=(b,v,\gamma,\om): b\in\R^2,~\gamma\in\R,~ |v|\le \ov v, ~|\om|\le \ov\om\}$.
\end{defin}
Let  $T_{a,\theta}$ be the translation:
\[
 T_{a,\theta}:(A(\cdot),\Pi(\cdot),q,p,\varphi,M)\mapsto(A(\cdot-a),\Pi(\cdot-a),q+a,p,\varphi+\theta,M), \quad a\in\R^2,\,\,\theta\in\R.
 \]
 \begin{lemma}\la{skewpro}
Let \eqref{rosym} hold, $\al<3$, $0<\ov v<1$,  and $\ov\omega>0$.
Then
\\
i) there exists a neighborhood ${\cal O}_\al({\cal S})$ of ${\cal S}$ in ${\cal E}_\al$ and a
map ${\bf \Pi}:{\cal O}_\al({\cal S})\to{\cal S}$  such that ${\bf \Pi}$ is uniformly
continuous on ${\cal O}_\al({\cal S})\cap {\cal E}_\al(\ov v,\ov\om)$ in the metric of ${\cal E}_\al$,
\be\la{proj}
{\bf \Pi} Y=Y~~\mbox{for}~~ Y\in{\cal S}, ~~~~~\mbox{and}~~~~~Y-S \nmid {\cal T}_S{\cal S},~~\mbox{where}~~S={\bf \Pi} Y.
\ee
ii) ${\cal O}_\al({\cal S})$ is invariant with respect to the translations
 $T_{a,\theta}$, and
 ${\bf \Pi }T_{a,\theta}Y=T_{a,\theta}{\bf \Pi} Y$ for $Y\in{\cal O}_{\al}({\cal S})$.
\\
iii)  there exist $0<\ov v_1<1$   and $\ov\omega_1>0$, depending on $\ov v$ and $\ov\om$, such that   ${\bf \Pi} Y=S(\si)$ with  $\si\in \Sigma(\ov v_1,\ov\om_1)$ for
 $Y\in \cE_\al(\ov v,\ov\om)$.\\
iv) there exists  $z_\al(\ov v,\ov {\omega})>0$ such that  
$S(\si)+Z\in{\cal O}_\al({\cal S})$ if $\si\in\Sigma(\ov v_1,\ov\omega_1)$
and $\Vert Z\Vert_\al< z_\al(\ov v,\ov\om)$.
\end{lemma}
The proof is similar to that of  Lemma 3.4 in \cite{IKV2012}. 
\begin{cor}
We may assume $Y_0=S+Z_0$ where $S=S(\si_0)={\bf \Pi} Y_0$ and $Z_0$ satisfies   \eqref{close}. 
\end{cor}
\subsection{Linearization on the solitary manifold}\label{LSM}
As the first step in the proof of  Theorem \ref{main}, let us linearize the nonlinear system \eqref{mls3} at a solitary wave $S(\si(t))$.
We decompose  a solution to  \eqref{mls3} as  the sum
\be\la{dec}
Y(t)=S(\si(t))+Z(t),
\ee
where $\si(t)=(b(t),v(t),\gamma(t), \om(t))\in\Sigma$ is an arbitrary smooth function of $t\in\R$.
In more detail, denote $Y=(A,\Pi,q,p,\varphi, M)$ and $Z=(\Lambda,\Psi,r,\pi,\psi,\chi)$.
Then \eqref{dec} means that
\beqn\nonumber
A(x,t)&\!\!\!=\!\!\!&A_{v(t),\om(t)}(x\!-\!b(t))+\Lambda(x\!-\!b(t),t),\quad  \Pi(x,t)=\Pi_{v(t),\om(t)}(x\!-\!b(t))+\Psi(x\!-\!b(t),t),\\
\nonumber
q(t)&\!\!\!=\!\!\!&b(t)+r(t),\quad p(t)=p_{v(t),\om(t)}+\pi(t),\quad \varphi(t)=\gamma(t)+\psi(t),\quad M(t)=M_{v(t),\om(t)}+\chi(t).
\eeqn
Substituting this decomposition  into \eqref{mls3} and setting $y=x-b(t)$, we obtain 
\beqn\la{dot-A}
\!\!\!\!\!\!\!\!\!\!\!\!\!\!\!\!\!\!\!\!\!\!\!\!&&\dot A=(\dot v\cdot \!\na_v +\dot\om\pa_{\om})A_{v,\om}(y)-\dot b\cdot \na \big[A_{v,\om}(y)+\Lambda (y,t)\big]+\dot\Lambda(y,t)
=\Pi_{v,\om}(y)+\Psi(y,t),\\
\nonumber
\!\!\!\!\!\!\!\!\!\!\!\!\!\!\!\!\!\!\!\!\!\!\!\!&&\dot\Pi=(\dot v\cdot \na_v+\dot\om \pa_{\om})\Pi_{v,\om}(y)-\dot b \cdot\na\big[\Pi_{v,\om}(y)+\Psi(y,t)\big]+\dot\Psi(y,t)\\
\la{dot-Pi}
\!\!\!\!\!\!\!\!\!\!\!\!\!\!\!\!\!\!\!\!\!\!\!\!&&~~~~~~~=\De A_{v,\om}(y)+\De\Lambda(y,t)-\dot\varphi J\varrho(y-r)+{\cal P}[\dot q \rho(y-r)],\\
\la{dot-q}
\!\!\!\!\!\!\!\!\!\!\!\!\!\!\!\!\!\!\!\!\!\!\!\!&&m\dot q=m(\dot b+\dot r)=p_{v,\om}+\pi-\langle A_{v,\om}+\Lambda,\rho(y-r)\rangle,\\
\la{dot-p}
\!\!\!\!\!\!\!\!\!\!\!\!&&\!\!\!\!\!\dot p=(\dot v\cdot \!\na_v +\dot\om\pa_{\om})p_{v,\om}+\dot \pi(t)
=\dot\varphi\langle\na\cdot J(A_{v,\om}+\Lambda),\varrho(y\!-\!r)\rangle-\langle (A_{v,\om}+\Lambda)\cdot\!\dot q,\na\rho(y\!-\!r)\rangle,\\
\la{dot-phi}
\!\!\!\!\!\!\!\!\!\!\!\!&&\!\!\!\!\!I\dot \varphi(t)=I(\dot \gamma+\dot\psi)=M_{v,\om}+\chi+\langle A_{v,\om}+\Lambda, J\varrho(y-r)\rangle,\\
\la{dot-M}
\!\!\!\!\!\!\!\!\!\!\!\!&&\!\!\!\!\!\dot M=(\dot v\cdot \na_v +\dot \om\pa_{\om})M_{v,\om}+\dot\chi=0.
\eeqn
Equations  \eqref{Hs5}, \eqref{dot-A}  and \eqref{dot-M} imply
\beqn\nonumber
\dot \Lambda(y,t)&=&\Psi(y,t)+\dot b\cdot \na \Lambda(y,t)+
(\dot b-v)\cdot \na A_{v,\omega}(y)-\dot v\cdot \na_v A_{v,\om}(y)-\dot {\omega}\pa_{\omega} A_{v,\omega}(y),\\
\la{chi-eq}
 \dot\chi&=&-\dot v\cdot \na_v M_{v,\omega}-\dot {\omega}\pa_{\omega}M_{v,\omega}.
\eeqn
Further, equations \eqref{solvom}, \eqref{Hs5}, \eqref{dot-Pi}   and \eqref{dot-p}--\eqref{dot-M}  imply
\beqn\nonumber
\dot \Psi(y,t)&=&\De\Lambda(y,t)+\dot b\cdot \!\na\Psi(y,t)-\dot v\cdot \na_v \Pi_{v,\omega}(y)-\dot {\omega}\pa_{\omega} \Pi_{v,\omega}(y)+(\dot b\!-\!v)\!\cdot\! \na\Pi_{v,\omega}(y)\\
\nonumber
&-&\dot\varphi J\varrho(y-r)+{\omega} J\varrho(y)+{\cal P}[\dot q \rho(y-r)-v\rho(y)],\\
\nonumber
m\dot r&=&-m\dot b+mv+\langle A_{v,\om},\rho\rangle+\pi-\langle A_{v,\om}+\Lambda,\rho(y-r)\rangle,\\
\nonumber
\dot \pi&=&-\dot v\cdot \na_v p_{v,\omega}-\dot {\omega}\pa_{\omega}p_{v,\omega}-\langle (A_{v,\omega}+\Lambda)\cdot\dot q,\na\rho(y-r)\rangle
+\langle A_{v,\omega}\cdot v,\na\rho\rangle,\\
\nonumber
&+&\dot\varphi\langle\na\cdot J(A_{v,\omega}+\Lambda),\varrho(y-r)\rangle-{\omega}\langle \nabla\cdot JA_{v,\omega},\varrho\rangle,\\
\la{dot-psi}
I\dot\psi&=&-I(\dot\gamma-{\omega})  -\langle A_{v,\omega},J\varrho\rangle+\chi+\langle A_{v,\omega}+\Lambda, J\varrho(y-r)\rangle.
\eeqn
By \eqref{solit3},
\beqn\nonumber
&&\!\!\!\!\!\!\!\!\!\!\langle A_{v,\omega}, r\cdot\na\rho\rangle=-\om JPr,\quad \langle A_{v,\om}, (r\cdot\na) J\varrho\rangle=r\cdot PJv,\quad 
\langle A_{v,\omega}\cdot\pi,\na\rho\rangle=\om PJ\pi,\\
\nonumber
&&\!\!\!\!\!\!\!\!\!\!
\langle A_{v,\omega}\cdot\langle\Lambda,\rho\rangle,\na\rho\rangle={\omega} PJ\langle \Lam,\rho\rangle,\quad
\langle A_{v,\omega}\cdot  JPr,\na\rho\rangle ={\omega} P^2r,\\
\la{JAro}
&&\!\!\!\!\!\!\!\!\!\!\langle\na\cdot JA_{v,\omega},\varrho\rangle=PJ v,\quad 
\langle\na\cdot JA_{v,\omega},r\cdot\na\varrho\rangle={\omega} Fr,\quad 
\langle \na(A_{v,\omega}\cdot v), r\cdot\na\rho\rangle=Qr,
\eeqn
where $P$, $F$ and  $Q$ are symmetric  matrices  with the  entries
 \be\la{c-ro}
 P_{jl}= \int\! \frac{k_l\na_j\hat\rho(k)}{\hat D_0}\hat\rho(k) dk,\quad F_{jl}=\int\! \frac{k_jk_l|\na\hat\rho|^2}{\hat D_0} dk,
 \quad Q_{jl}=\int\! \big(v^2k^2-(v\cdot k)^2\big)\frac{k_lk_j|\hat\rho(k)|^2dk}{k^2\hat D_0}.
\ee
Using  \eqref{JAro},  we rewrite  \eqref{dot-psi} as
\beqn\nonumber
\dot \Psi&\!\!=\!\!&\De\Lambda+\dot b\cdot \!\na\Psi\!+(\dot b\!-\!v)\!\cdot\! \na\Pi_{v,\om}\!
-\dot v\cdot \na_v \Pi_{v,\om}-\dot\om\pa_{\om} \Pi_{v,\om}+\om (r \cdot\na) J\varrho\\
\nonumber
&\!\!+\!\!&{\cal P}\big[\frac{\rho}m\big(\pi-\langle\Lam,\rho\rangle-\om JPr\big)-vr \cdot\na\rho\big]
-\frac {J\varrho}{I}\big[\chi+\langle\Lam,J\varrho\rangle-r\cdot PJv\big]+N_2(v, {\omega}, Z),\\ 
\nonumber
m\dot r&=&m(v-\dot b)+\pi-\langle \Lambda,\rho\rangle-\om JPr+N_3(v, {\omega}, Z),\\
\nonumber
\dot \pi&\!\!=\!\!&-\dot v\cdot \na_v p_{v,\omega}-\dot {\omega}\pa_{\omega}p_{v,\omega}
-\langle v\cdot\Lambda,\na\rho\rangle- \frac{\om}m PJ(\pi-\langle\Lam,\rho\rangle)+{\omega}\langle \na\cdot J\Lam,\varrho\rangle\\
\nonumber  
&\!\!-\!\!&(\frac {\omega^2 }m P^2+Q+{\omega}^2 F)r-\frac 1{I}(r\cdot PJv) PJv+ \frac{1}{I}(\chi+\langle\Lam,J\varrho\rangle) PJv+N_4(v, {\omega}, Z),\\
\la{psi-eq}
I\dot \psi
 &\!\!=\!\!&-I(\dot\gamma-\om)+\chi+\langle\Lam,J\varrho\rangle-r\cdot PJv+N_5(v, \om, Z), 
\eeqn
where $\Vert N_{j}(v, {\omega}, Z)\Vert_\beta\le C(\ov v,\ov\om) \Vert Z\Vert_{-\beta}^2$, $j=2,\dots, 5$, 
uniformly in $v$, $\omega$, and $Z$ with   $|v|\le \ov v<1$,  $|\om|\le\ov\om$, and  $\Vert Z\Vert_{-\beta}< z_{-\beta}(\ov v,\ov\om)$.
Combining equations \eqref{chi-eq}  and \eqref{psi-eq}, we obtain
\be\la{lin}
\dot Z(t)=L_{v(t), \dot b(t),{\omega}(t)}Z(t)+T(\sigma(t))+N(v(t),{\omega}(t),Z(t)),\,\,\,t\in\R, \quad Z=(\Lam,\Psi,r,\pi,\psi,\chi).
\ee
Here 
\be\la{AA}
L_{v,u,{\omega}}:=\left(
\ba{cccccc}
u \cdot\na & 1 & 0 & 0 &0 &0\\
\De -{\cal P}[\frac{\langle\cdot,\rho\rangle\rho}{m}] -\frac{\langle \cdot,J\varrho\rangle J\varrho}{I}& u \cdot \na & {B}
&{\cal P} [\frac{\rho}m\cdot] & 0&-\frac{J\varrho}I\\
-\frac 1m\langle\cdot,\rho\rangle & 0 & -\frac{\omega}{m}JP & \frac 1m&0&0\\
  {B}^* & 0 &-G-(\cdot PJv) \frac{PJv}{I} & -\frac{\omega}mPJ& 0&\frac{1}I PJv \\
\frac{1}{I}\langle \cdot,J\varrho\rangle &0 &-\frac{1}I (PJv\cdot)&0&0&\frac{1}{I}\\
0&0&0&0&0&0
\ea\right),
\ee
\be\la{T}
T(\sigma)=\begin{pmatrix}
(u-v)\cdot \na A_{v,\omega}-\dot v\cdot \na_v A_{v,\omega}-\dot {\omega}\pa_{\om} A_{v,\omega}\\
(u-v)\!\cdot\na\Pi_{v,\omega}-\dot v\cdot \na_v \Pi_{v,\omega}-\dot {\omega}\pa_{\omega} \Pi_{v,\omega}\\
v-u \\
-\dot v\cdot \na_v p_{v,\omega}-\dot {\omega}\pa_{\omega}p_{v,\omega}\\
-(\dot\gamma-\om)\\
-\dot v\cdot \na_v M_{v,\omega}-\dot {\omega}\pa_{\omega}M_{v,\omega}
\end{pmatrix}\!,\quad
N(v, {\omega},Z)=\begin{pmatrix}
0 \\ N_2(v, {\omega},Z)  \\\ N_3(v, {\omega},Z) \\ N_4 (v, {\omega},Z)\\  N_5 (v, {\omega},Z)\\0
\end{pmatrix}.
\ee
The operators $B=B(v,{\omega})$, $B^*=B^*(v,{\omega})$ and $G=G(v,{\omega})$  act as follows
\beqn \la{B-op}
 {B}r&\!\!=\!\!&{\cal P}[-\frac{\omega}{m}\rho JPr-v(r\cdot \na\rho)]+{\omega}(r\cdot \na)J\varrho+\frac{1}{I}(r\cdot PJv)J\varrho,\\
 \la{B1-op}
  {B}^*\Lam&\!\!=\!\!&\frac {\omega}m PJ\langle \Lam,\rho\rangle-\langle v\cdot\Lam,\na\rho\rangle+{\omega}\langle \na\cdot J\Lam,\varrho\rangle
  +\frac{1}{I}PJv\langle\Lam,J\varrho\rangle,\\
 \la{j-al}
Gr&\!\!=\!\!&(\frac {\omega^2 }m P^2+Q+{\omega}^2 F)r.
\eeqn
Furthermore,   
\be\la{N-est}
\Vert N(v, {\omega}, Z)\Vert_{\beta}\le C(\ov v,\ov\om) \Vert Z\Vert_{-\beta}^2,
\ee
uniformly in $v$, $\omega$, and $Z$ with  $|v|\le \ov v<1$, $|\om|\le\ov\om$, and $\Vert Z\Vert_{-\beta}< z_{-\beta}(\ov v,\ov\om)$. 
\subsection{The linearized equation}\label{LinEq}
Here we collect some properties  of the linearized equation
\be\la{line}
\dot X(t)=L_{v,u,\omega}X(t),\quad v\in {\cal V}, \quad u\in \R^2,\quad \omega\in\R,\quad t\in\R.
\ee
\begin{lemma} \la{haml}
i) Equation \eqref{line}  can be written as the Hamiltonian system
\be\la{lineh}
\dot X(t)= \bJ{\cal D}{\cal H}_{v,u,\omega}(X(t)),\quad X=(\Lam,\Psi, r,\pi,\psi,\chi),\quad t\in\R,
\ee
where 
\beqn\nonumber
{\cal H}_{v,u,\omega}(X)&\!\!=\!\!&\frac12\int\Big[|\Psi|^2+|\na\Lambda|^2\Big]dy+\int\Psi (u\cdot\na)\Lambda dy
+\frac{1}{2m}\langle\Lambda,\rho\rangle^2+\frac {1}{2I}\langle \Lambda,J\varrho\rangle^2\\
\nonumber
&\!\!+\!\!&\int (\frac{\omega}{m}JP r-\frac{\pi}m)\cdot\Lambda\rho dy+\frac{1}{I}(\chi-r\cdot PJv)\int \Lambda\cdot J\varrho dy+\int(\na\rho\cdot r)v\cdot \Lambda dy\\
\nonumber
&\!\!-\!\!&{\omega}\!\int\!(r\!\cdot\!\na)\Lam\cdot\! J\varrho dy +\frac{\chi^2}{2I}+\frac{\pi^2}{2m} -\frac{\omega}{m} JPr\! \cdot\pi
+\frac 12 r\cdot\! Gr+\frac{1}{I}(r\!\cdot\! PJv)^2 +\frac{\chi}I r\cdot\! PJv.
\la{H0}
\eeqn
ii) The energy conservation law holds for  solutions $X(t)\in C^1(\R,\cE)$:
\be\la{enec}
{\cal H}_{v,u,\omega}(X(t))={\rm const},~~~~~t\in\R.
\ee
iii) The skew-symmetry relation holds:
\be\la{com}
\Omega(L_{v,u,\omega}X_1,X_2)=-\Omega(X_1,L_{v,u,\omega}X_2), ~~~~~~~~X_1,X_2\in \cE.
\ee
\end{lemma}
\begin{proof}
 i) Equality \eqref{lineh} is easily verified by differentiation.
\\
ii) The energy conservation law
follows from  \eqref{lineh} and the chain rule   for Fr\'echet derivatives:  
\[
\frac d{dt}\cH_{v,u,\om}
(X(t))=\langle {\cal D}\cH_{v,u,\omega}(X(t)),\dot X(t)\rangle=
\langle {\cal D}\cH_{v,u,\omega}(X(t)),\bJ {\cal D}\cH_{v,u,\omega}(X(t))\rangle=0,~~~~~~t\in\R.
\]
iii) The skew-symmetry holds
since $L_{v,u,\omega}X=\bJ{\cal D}\cH_{v,u,\omega}(X)$, and the linear operator\\
$X\mapsto {\cal D}\cH_{v,u,\omega}(X)$ is symmetric, being the Fr\'echet derivative of a quadratic form.
\end{proof}
\begin{lemma} \la{ceig1}
The operator $L_{v,u,\omega}$ acts on the tangent vectors $\tau_j(v,\omega)$  as follows:
\beqn\la{Atanform}
L_{v,u,\omega}[\tau_j(v,\omega)]&=&(u-v)\cdot\na\tau_j(v,{\omega}),\qquad j=1,2,\\
\la{Atanform2}
L_{v,u,\omega}[\tau_{j+2}(v,\om)]&=&(u-v)\cdot\na\tau_{j+2}^{-}(v,{\omega})+\tau_j(v,{\omega}),\qquad j=1,2,\\
\la{Atanform1}
L_{v,u,\omega}[\tau_5(v,\omega)]&=&0,
\quad L_{v,u,\omega}[\tau_{6}(v,\omega)]=(u-v)\cdot\na\tau_{6}^{-}(v,{\omega})+\tau_5(v,\omega).
\eeqn
Here 
$\tau_{j+2}^{-}(v,{\omega})=(\pa_{v_j}A_{v,{\omega}},\,\pa_{v_j}\Pi_{v,\omega},\, 0,\, \pa_{v_j}p_{v,\omega},\,   0,\,  0)$ and 
$\tau_{6}^{-}(v,{\omega})=(\pa_{\omega}A_{v,\omega},\, \pa_{\omega}\Pi_{v,\omega},\, 0,\, 0,\, 0,\, 0)$.
\end{lemma}
\begin{proof}
We will  write  $A=A_{v,\omega}$, $\Pi=\Pi_{v,\omega}$, etc.  
Differentiating  \eqref{Hs5} with respect to $b_j$ and $v_j$, we obtain 
\beqn\la{d1}
\!\!\!\!\!\!\!\!-(v\cdot\na)\pa_j A&=&\pa_j\Pi,\qquad
-(v\cdot\na)\pa_j\Pi=\Delta\pa_j A+{\cal P}[v\pa_j\rho] -{\omega}\pa_j J\varrho,\\
\la{d2}
\!\!\!\!\!\!\!\!-\pa_j A-(v\cdot\na)\pa_{v_j} A&=&\pa_{v_j}\Pi,\qquad
-\pa_j\Pi- (v\cdot\na)\pa_{v_j}\Pi=\De\pa_{v_j} A+{\cal P}[e_j\rho],\\
\la{d-om}
\!\!\!\!\!\!\!\!-(v\cdot\na)\pa_\om A&=&\pa_{\om}\Pi,\qquad-(v\cdot\na)\pa_\om \Pi=\De \pa_{\om }A-J\varrho.
\eeqn
{\it Step i)}  First, we prove  \eqref {Atanform}.
By definition \eqref{AA} of  $L_{v,u,\omega}$,   it is sufficient to check that
\beqn\la{e1}
&&\!\!\!\!\!\!\!\!\!\!\!\!\!\!\!\!\!\!\!\!\!\!\!\!\!\pa_j \Pi+(u\cdot\nabla) \pa_j A=((u-v)\cdot\na) \pa_j A,\\
\nonumber
&&\!\!\!\!\!\!\!\!\!\!\!\!\!\!\!\!\!\!\!\!\!\!\!\!\!\Delta \pa_j A\!-\!{\cal P}\big[\frac{\langle\pa_j A,\rho\rangle\rho}{m}-\frac{{\omega}\rho}{m}JPe_j\!+\!v\pa_j\rho\big]\!
-\!\big[\frac{\langle\pa_j A,J\varrho\rangle}{I}+\!\frac{e_j\!\cdot \!PJv}{I}+{\omega}\pa_j\big]J\varrho+(v\cdot\na)\pa_j\Pi\\
\la{e2}
&&=((u-v)\cdot\na) \pa_j \Pi,\\
\la{e3}
&&\!\!\!\!\!\!\!\!\!\!\!\!\!\!\!\!\!\!\!\!\!\!\!\!\!\langle \pa_j A,\rho\rangle-{\omega} JPe_j= 0,\qquad\quad\langle \pa_j A,J\varrho\rangle+(PJv)_j=0,\\
\la{e4}
&&\!\!\!\!\!\!\!\!\!\!\!\!\!\!\!\!\!\!\!\!\!\!\!\!\!\langle v\cdot\pa_j A,\na\rho\rangle-\frac {\omega}m\langle  PJ\pa_j A,\rho\rangle
-{\omega}\langle \na\cdot J\pa_j A,\varrho\rangle -\frac{1}{I}PJv\langle\pa_j A,J\varrho\rangle +Ge_j=0.
\eeqn
Equation \eqref{e1}  follows directly from \eqref{d1}.
Further,  \eqref{solit3} and  \eqref{c-ro} imply  that
\be\la{d4}
\langle \pa_j A,\rho\rangle={\omega}\langle \frac{k_jJ\na\hat\rho}{\hat D_0},\hat\rho\rangle = {\omega} JPe_j,\quad
\langle \pa_j A,J\varrho\rangle=\langle \frac{k_j\hat\rho}{\hat D_0}v,J\na\hat\rho\rangle=-e_j\cdot PJv.
\ee
Then \eqref{e3}  follows.
 Moreover,  \eqref{d1} and \eqref{d4} imply \eqref{e2}.
It remains to prove  \eqref{e4}. By \eqref{solit3} and \eqref{c-ro},
\[
 \langle PJ\pa_j A,\rho\rangle=-{\omega} P^2e_j,\quad  
\langle v\cdot\pa_j A,\na\rho\rangle= Qe_j,\quad 
\langle \na\cdot J\pa_j A,\varrho\rangle=-{\omega} Fe_j.
\]
Hence, \eqref{e4} follows from \eqref{e3} and definition \eqref{j-al} of the operator $G$. 
 \smallskip\\
{\it Step ii)}  Now we prove \eqref {Atanform2}. By \eqref{AA}, it is sufficient to check that 
\beqn\la{e21}
&&\!\!\!\!\!\!\!\!\!\!\!\!\!\!\!\!\!\!\!\!\pa_{v_j}\Pi+(u\cdot\na) \pa_{v_j} A+\pa_j A=((u-v) \cdot\na) \pa_{v_j} A,\\
\nonumber
&&\!\!\!\!\!\!\!\!\!\!\!\!\!\!\!\!\!\!\!\!\De\pa_{v_j} A-{\cal P}\big[\frac{\rho}{m}\big(\langle\pa_{v_j} A,\rho\rangle-\pa_{v_j}p\big)\big]
-\frac{1}{I}\big[\langle \pa_{v_j }A,J\varrho\rangle+\pa_{v_j}M\big]J\varrho+(u\cdot\na) \pa_{v_j} \Pi+\pa_j \Pi\\
\la{e22}
&&=((u-v)\cdot\na) \pa_{v_j} \Pi,\\
\la{e23}
&&\!\!\!\!\!\!\!\!\!\!\!\!\!\!\!\!\!\!\!\!-\pa_{v_j}\langle A,\rho\rangle+\pa_{v_j}p=me_j,\qquad\quad\langle \pa_{v_j} A,J\varrho\rangle+\pa_{v_j}M =0,\\
\la{e24}
&&\!\!\!\!\!\!\!\!\!\!\!\!\!\!\!\!\!\!\!\!v\cdot\!\langle \pa_{v_j} A,\!\na\rho\rangle\! -\!\frac {\omega}m PJ\big[\langle \pa_{v_j}A,\rho\rangle\!-\!\pa_{v_j}p\big]
\!+\!{\omega}\langle\na\!\cdot\! J\pa_{v_j}A,\varrho\rangle\!-\!\frac{1}{I}PJv\big[\pa_{v_j}M\!+\!\langle\pa_{v_j}A,J\varrho\rangle\!\big]\!=\!0.
\eeqn
Equations \eqref{e21} and  \eqref{e23}  follow directly from \eqref{d2}, \eqref{solY}, and \eqref{solvom}. 
Equation \eqref{e22} follows from \eqref{d2} and  \eqref{e23}.
It remains to prove  \eqref{e24}. Using  \eqref{e23}, we rewrite \eqref{e24} as
\be\la{e25}
v\cdot\pa_{v_j}\langle  A,\!\na\rho\rangle-{\omega}\pa_{v_j}\langle\na\!\cdot\! JA,\varrho\rangle+{\omega} PJe_j=0.
\ee
By \eqref{JAro},
$\pa_{v_j}\langle\na\cdot JA,\varrho\rangle=\pa_{v_j} (PJv)=PJe_j+(\pa_{v_j} P)Jv$.
Further, \eqref{c-ro} implies that
\be\la{ee5}
\langle A_{\ell},\na\rho\rangle=-{\omega}\langle \frac{(J\na\hat\rho)_{\ell}}{\hat D_0},k\hat\rho\rangle={\omega}PJe_{\ell}.
\ee
Hence, $v\cdot\pa_{v_j} \langle A,\na\rho\rangle={\omega} (\pa_{v_j} P)Jv$, and   \eqref{e25} follows. 
 \smallskip\\
{\it Step iii)}  It remains to prove \eqref{Atanform1}. 
The first equation  of \eqref{Atanform1} is trivial  by \eqref{inb} and \eqref{AA}. 
Due to \eqref{AA}, the second   equation  of \eqref{Atanform1}  is equivalent to the following  system 
\beqn\la{e31}
&&\!\!\!\!\!\!\!\!\!\!\!\!\!\!\!\!\!\!\!\!\!\!\!\pa_{\omega}\Pi+(u\cdot\na)\pa_{\omega}A=((u-v)\cdot\na)\pa_{\omega}A,\\
\la{e32}
&&\!\!\!\!\!\!\!\!\!\!\!\!\!\!\!\!\!\!\!\!\!\!\!\!\De\pa_{\omega}A -{\cal P}\big[\frac{\langle\pa_{\omega} A,\rho\rangle\rho}{m}\big]
\frac {1}{I}\big[\langle \pa_{\omega}A, J\varrho\rangle+\pa_{\omega}M\big] J\varrho
+(u\cdot\na)\pa_{\omega}\Pi= ((u-v)\cdot\na)\pa_{\omega}\Pi,\\
\la{e33}
&&\!\!\!\!\!\!\!\!\!\!\!\!\!\!\!\!\!\!\!\!\!\!\!\!\langle \pa_{\omega} A,\rho\rangle =0,\qquad -\langle \pa_{\omega}A, J\varrho\rangle+ \pa_{\omega}M=I,\\
\la{e34}
&&\!\!\!\!\!\!\!\!\!\!\!\!\!\!\!\!\!\!\!\!\!\!\!\!-v\cdot\langle \pa_{\omega}A,\na\rho\rangle+ \frac{\omega}{m}PJ\langle\pa_{\omega}A,\rho\rangle +{\omega}\langle\na\cdot J\pa_{\omega}A,\varrho\rangle
+\frac{1}{I}PJv\big[\langle \pa_{\omega}A, J\varrho\rangle+\pa_{\omega}M)\big]=0.
\eeqn
Equation  \eqref{e31} follows directly from the first equation of \eqref{d-om}.  
Equations \eqref{e33} follow from \eqref{solvom} and \eqref{solit3}, and  
\eqref{e32} follows from \eqref{d-om} using  \eqref{e33}. 
Finally,  \eqref{e34}  follows from  \eqref{solit3}, \eqref{ee5}, and  \eqref{e33}.
\end{proof}
We will apply Lemmas \ref{haml} and  \ref{ceig1} mainly to the operator $L_{v,\om}:=L_{v,v,\omega}$
corresponding to $u=v$.
In that case, the linearized equation has the following additional essential features.
\begin{lemma}\la{ceig2}
{\it i)} The tangent vectors $\tau_1(v,\omega)$, $\tau_2(v,\omega)$, and  $\tau_5(v,\omega)$ are eigenvectors,
and $\tau_{3}(v,\omega)$, $\tau_{4}(v,\omega)$, and $\tau_6(v,\omega)$ are root vectors of the
operator $L_{v,\omega}$, corresponding to the  zero eigenvalue, i.e.,
\beqn\la{Atanformv}
L_{v,\omega}[\tau_j(v,{\omega})]&=&0, \quad   L_{v,{\omega}}[\tau_{j+2}(v,{\omega})]=\tau_j(v,{\omega}),\quad j=1,2\\
\la{Atanformv1}
 L_{v,\omega}[\tau_5(v,{\omega})]&=&0,\quad  L_{v,\omega}[\tau_6(v,{\omega})]=\tau_5(v,\omega).
\eeqn
{\it ii)} The Hamiltonian  functional ${\cal H}_{v,\omega}(X):={\cal H}_{v,v,\omega}(X)$ is nonnegative definite,
\be\la{H+}
{\cal H}_{v,\omega}(X)\ge 0,\quad X\in {\cal E}.
\ee
 \end{lemma}
\begin{proof}
Statement {\it i)}  follows from \eqref{Atanform}-\eqref{Atanform1} with $u = v$. Let us prove  statement {\it ii)}. 
We assume that $v = (|v|, 0)$ and represent  ${\cal H}_{v,\omega}(X)$ as
\[
{\cal H}_{v,\omega}(X)=\frac12 \Vert\Psi+(v\cdot\na)\Lam\Vert_{L^2(\R^2)}^2+\frac {1}{2I}(\chi-r\cdot PJv+\langle \Lambda,J\varrho\rangle)^2
+\frac{1}{2m}(\pi\!-\!\langle\Lam,\rho\rangle\!-\!{\omega} JP r)^2+h(X),
\]
\[
h(X)=\frac 12\langle(-\Delta+(v\cdot\na)^2)\Lam,\Lam\rangle
+\!\int\!\big((\na\rho\cdot r)v\cdot \Lambda-{\omega}(r\cdot\na)\Lam\cdot J\varrho\big) dy
+\sum\limits_{j}\frac {v^2Q_{jj}+{\omega}^2 F_{jj}}{2}r_j^2.
\]
It suffices to prove that $h(X)\ge 0$. In the  Fourier representation, 
\beqn\nonumber
h(X)&=&\frac 12\int\big(\hat D_0|\rIm\hat\Lam|^2+2(k\cdot r)\hat\rho |v|\rIm\hat\Lam_1\big)dk+\frac{v^2}2\sum\limits_{j}Q_{jj}r_j^2\\
\nonumber
&+&\frac 12\int\big(\hat D_0|\rRe\hat\Lam|^2-2{\omega}(k\cdot r)(\rRe\hat\Lam\cdot J\na\rho) \big)dk
+\frac{{\omega}^2}2 \sum\limits_{j}F_{jj}r_j^2=\frac 12(h_1(X)+h_2(X)).
\eeqn
Using  \eqref{c-ro}, we obtain
\beqn\nonumber
h_1(X)&\!\!\!=\!\!\!&\int\Big[\hat D_0|\rIm\hat\Lam|^2+2(k\cdot r)\hat\rho |v|(e_1-\frac{k_1k}{k^2})\cdot\rIm\hat\Lam
+v^2\frac{(k\cdot r)^2|\hat\rho|^2}{\hat D_0}(e_1-\frac{k_1k}{k^2})^2\Big]dk\\
\nonumber
&\!\!\!=\!\!\!&\int \Big[\hat D_0^{1/2}\rIm\hat\Lam+\frac{(k\cdot r)|v|\hat\rho}{\hat D_0^{1/2}}(e_1-\frac{k_1k}{k^2})\Big]^2dk\ge 0,\\
\nonumber
h_2(X)&\!\!\!=\!\!\!&\int\Big[\hat D_0|\rRe\hat\Lam|^2\!-{2\omega} (k\cdot r)(\rRe\hat\Lam_1\na_2\hat\rho-\rRe\hat\Lam_2\na_1\hat\rho)
+{\omega}^2\frac{(k\cdot r)^2}{\hat D_0}\big((\na_1\hat\rho)^2+(\na_2\hat\rho)^2\big)\Big]dk\\
\nonumber
&\!\!\!=\!\!\!&\int\Big[\hat D_0^{1/2}\rRe\hat\Lam_1-\frac{{\omega}(k\cdot r)\na_2\hat\rho}{\hat D_0^{1/2}}\Big]^2dk+
\int\Big[\hat D_0^{1/2}\rRe\hat\Lam_2+\frac{{\omega}(k\cdot r)\na_1\hat\rho}{\hat D_0^{1/2}}\Big]^2dk\ge 0.
\eeqn
\end{proof}
\subsection{Symplectic decomposition of the dynamics}\label{SymDec}
Here we decompose the dynamics into two components: along the manifold ${\cal S}$ and in transversal directions. 
The equation \eqref{lin} is obtained without any assumption on $\sigma(t)$ in \eqref{dec}. 
Now we are going to choose $S(\si(t))={\bf \Pi} Y(t)$. This is possible for $t=0$ by \eqref{close},  
so $S(\si(0))={\bf\Pi} Y(0)$ and  $Z(0)=Y(0)-S(\si(0))$ are well defined.
Moreover, Lemma \ref{skewpro}  implies  that $S(\si(t))={\bf\Pi} Y(t)$ and  $Z(t)=Y(t)-S(\si(t))$ are well defined for $t\ge 0$ as long as
 $\Vert Z(t)\Vert_{-\beta} < z_{-\beta}(\ov v,\ov\om)$. This is formalized by the following  definition.
\begin{defin}
$t_*$ is the ``exit time'':
\[
t_*=\sup \{t>0: \Vert Z(s)\Vert_{-\beta} < z_{-\beta}(\ov v,\ov\om),~~0\le s\le t\}.
\]
\end{defin}
For $0<t<t_*$, we  set $S(\si(t))={\bf \Pi} Y(t)$, which is equivalent to 
\be\la{orth}
\Om(Z(t),\tau_j(t))=0,\quad \tau_j(t)=\tau_j(v(t),{\omega}(t)),\quad j=1,\dots,6, ~~~~~~~0\le t<t_*. 
\ee
It would be convenient  to use some other parameters $\ti\sigma=(c,v,\vartheta,\omega)$ instead of $\si=(b,v,\gamma,\omega)$, where
$c(t)=b(t)-\ds\int^t_0 v(\tau)d\tau$ and $\vartheta(t)=\gamma(t)-\ds\int^t_0 \omega(\tau)d\tau$. One has
\be\la{vw}
\dot c(t)=\dot b(t)-v(t)=u(t)-v(t), \quad \dot\vartheta(t)=\dot\gamma(t)-{\omega}(t), \quad 0\le t<t_*.
\ee
\begin{lemma}\la{mod}
Let $Y(t)$ be a solution to the Cauchy problem for \eqref{mls3}, and let \eqref{dec}, \eqref{orth} hold. 
Then 
\be\la{parameq}
\dot{\ti\si}(t)=(\dot c(t), \dot v(t), \dot \vartheta(t), \dot{\omega}(t)) ={\cal O}(\Vert Z(t)\Vert_{-\beta}^2), ~~~~~~~0\le t<t_*
\ee
uniformly in $v$, $\omega$ and $Z$ with $|v|\le \ov v<1$,  $|\om|\le\ov\om$ and $\Vert Z\Vert_{-\beta}< z_{-\beta}(\ov v,\ov\om)$. 
\end{lemma}
\begin{proof} 
Differentiating  the orthogonality conditions  \eqref{orth}  with respect to  $t$, we obtain
\be\la{ortder}
0=\Omega(\dot Z,\tau_j)+\Omega(Z,\dot\tau_j)=\Omega(L_{v,u,\omega}Z+T+N,\tau_j)+\Omega(Z,\dot\tau_j), \quad j=1,\dots,6,\quad 0\le t<t_*.
\ee
By  \eqref{inb} and \eqref{T},
$T(t)=-\sum\limits_{l=1}^2[\dot c_l\tau_l+\dot v_l\tau_{l+2}]-\dot\vartheta\tau_5-\dot\omega\tau_6$.
Hence, $\Om(T,\tau)=\Om(v,\omega)\dot{\ti\si}$.\\
Further,  $\Omega(L_{v,u,\omega}Z,\tau_j)=-\Omega(Z,L_{v,u,\omega}\tau_j)$ by the skew-symmetry \eqref{com}.  
Hence,    \eqref{Atanform}--\eqref{Atanform1} imply
\beqn\nonumber
\Omega(L_{v,u,\omega}Z,\tau_j)&\!\!\!=\!\!\!&-\Omega(Z,\dot c\cdot\na\tau_j),\qquad
\Omega(L_{v,u,\omega}Z,\tau_{j+2})=-\Omega(Z,\dot c\cdot\na\tau^{-}_{j+2}),\quad j=1,2,\\
\nonumber
\Omega(L_{v,u,\omega}Z,\tau_{5})&\!\!\!=\!\!&0,~~\qquad\qquad\qquad\qquad
\Omega(L_{v,u,\omega}Z,\tau_{6})=-\Omega(Z,\dot c\cdot\na \tau^{-}_{6})
\eeqn
since $\Om(Z,\tau_j)=0$.   Finally, 
$\Omega(Z,\dot\tau_j)= \Om(Z,\dot v\cdot\na_v\tau_j+\dot{\omega}\pa_{\omega}\tau_j)$.
As a result,  \eqref{ortder} becomes
\be\la{modul}
0=\Om(v,\omega)\dot{\ti\si}+{\cal M}(\si,Z)\dot{\ti\si}+{\cal N}(\si,Z),
\ee
where  ${\cal M}(\sigma,Z)={\cal O}(\Vert Z\Vert_{-\beta})$, and ${\cal N}(\sigma,Z)={\cal O}\big(\Vert Z\Vert_{-\beta}^2\big)$
uniformly in $|v|\le\ov v<1$, $|{\omega}|\le\ov\om$ and  $\Vert Z\Vert_{-\beta}< z_{-\beta}(\ov v,\ov\om)$.
Since $\Om(v,\omega)$ is invertible by Lemma \ref{Ome}, and $\Vert Z\Vert_{-\beta}$ is small, we can resolve \eqref{modul}
with respect to $\dot{\ti\si}$ and obtain  \eqref{parameq}.
\end{proof}
Formula \eqref{T} and Lemma \ref{mod}  imply that
$\Vert T(t)\Vert_{\beta}\le\! C(\ov v,\ov\om)\Vert Z\Vert_{-\beta}^2$. 
Now we rewrite \eqref{lin} as
\be\la{reduced}
\dot Z(t)=L_{v(t),u(t),{\omega}(t)}Z(t)+N^{+}(t), \quad  N^{+}(t):=T(\sigma(t))+N(v(t),{\omega}(t),Z(t)),\quad 0\le t<t_*,
\ee
where
\be\la{redN}
\Vert N^{+}(t)\Vert_{\beta}\le C(\ov v,\ov\om)\Vert Z\Vert_{-\beta}^2,\qquad 0\le t<t_*.
\ee
\subsection{Frozen transversal dynamics}\label{Frozen}
The linear part of  equation \eqref{reduced} is non-autonomous, hence we cannot  directly apply known methods of scattering theory. 
Similarly to the approach of \cite{BKKS, BS, IKV2006}, we reduce the problem to the analysis of the frozen linear equation.
Namely, we fix an  arbitrary $t_1\in [0,t_*)$ and rewrite  equation \eqref{reduced} in  ``frozen form''
\be\la{froz}
\dot Z(t)=L_1Z(t)+(L_{v(t),u(t),{\omega}(t)}-L_1)Z(t)+N^{+}(t),\,\,\,~~~~0\le t<t_*,
\ee
where $L_1=L_{v(t_1),v(t_1), {\omega}(t_1)}$. The following trick allows us to get rid of  the ``bad terms''  $[u(t)-v(t_1)] \cdot\na$ 
in the operator $\mathfrak{L}(t,t_1):=L_{v(t),u(t),{\omega}(t)} -L_1$. We denote
\be\la{dd1}
\ell_1(t):=\int_{t_1}^t(u(s)-v(t_1))ds, ~~~~0\le t\le t_1,
\ee
 and  change   variables $(y,t)\mapsto (y_1,t)=(y+\ell_1(t),t)$.
Next we define 
\beqn\la{Z1}
\check\Lam(y_1,t):=\Lam(y,t)=\Lam(y_1-\ell_1(t),t),\quad \check\Psi(y_1,t):=\Psi(y,t)=\Psi(y_1-\ell_1(t),t). 
\eeqn
Now  for $0\le t\le t_1$  we obtain  the ``frozen''  transversal dynamics
\be\la{redy1}
\frac{d}{dt}\check Z(t)=L_1\check Z(t)+\mathfrak{L}(t,t_1)\check Z(t)+\check N^+(t,t_1),\quad 
\check Z(t)=(\check\Lam(t),\check\Psi(t),r(t),\pi(t),\psi(t),\chi(t)).
\ee
Here $\check N^+(t,t_1)$ is $N^+(t)$ expressed in terms of $y=y_1-\ell_1(t)$,
and the nonzero elements of the matrix operator $\mathfrak{L}=\mathfrak{L}(t,t_1)$  act on the corresponding components of $\check Z$ 
as follows (cf. formulas  \eqref{AA} and \eqref{B-op}-\eqref{j-al}): 
\beqn\nonumber
\mathfrak{L}_{41}\check\Lam&\!\!\!=\!\!\!&\frac{1}{m}\delta(\om P)J\langle\check\Lambda,\rho\rangle
+\delta(\om)\langle\na\cdot J\check\Lambda,\varrho\rangle
- \langle \delta (v)\cdot\check\Lam,\na\rho\rangle+\frac{1}{I}\delta(PJ v)\langle\check\Lam,J\varrho\rangle,\\
\nonumber
\mathfrak{L}_{23}r&\!\!\!=\!\!\!&-{\cal P}[\frac{\rho}{m}\delta(\om JP)r+\delta (v)(r\cdot\na\rho)]+\delta(\om)(r\cdot\na)J\varrho
+\frac{1}{I}(r\cdot \delta(PJ v))J\varrho,\\
\nonumber
\mathfrak{L}_{33}r&\!\!\!=\!\!\!&-\frac{1}{m}\delta(\om JP)r, \qquad \mathfrak{L}_{44}\pi=-\frac{1}{m}\delta(\om P)J\pi,\\
\nonumber 
\mathfrak{L}_{43}r&\!\!\!=\!\!\!&-\big[\frac{\delta(\om^2 P^2)}m+\delta(Q)+\delta(\om^2 F)\big]r
-\frac{1}{I}\big[(r\!\cdot\! P(t)Jv(t))\delta(PJv)+(r\!\cdot\! \delta(PJv))P(t_1)Jv(t_1)\big]\\
\mathfrak{L}_{46}\chi&\!\!\!=\!\!\!&\frac{1}{I}\delta(PJv)\chi,\qquad \mathfrak{L}_{53}r =\frac{1}{I}\delta(PJv)\cdot r
\nonumber
\eeqn
where we denote $\delta(f)=\delta (f,t,t_1)=f(t)-f(t_1)$. It is easy to show that 
\be\la{bB-est}
\Vert \mathfrak{L}(t,t_1)\check Z\Vert_{\beta}\le C(\ov v,\ov\om)\Vert \check Z\Vert_{-\beta}\big[|v(t)-v(t_1)|+|\omega(t)-\omega(t_1)|\big].
\ee
Denote $\ov \ell_1(s):=\sup_{0\le t\le s} |\ell_1(t)| $, $0\le s\le t_1$, and reduce the exit time:
\[
t_*'=\sup \{t\in[0,t_*):
\ov \ell_1(s)\le 1,~~0\le s\le t\}.
\]
Evidently, $\Vert \check Z(t)\Vert_{\al}\le C(\al)\Vert Z(t)\Vert_{\al}$ for $t\le t_1< t_*'$ and any $\al\in\R$. Hence,  \eqref{parameq} and \eqref{bB-est} imply
\be\la{B1Z1est}
\Vert  \mathfrak{L}(t,t_1)\check Z(t)\Vert_{\beta}\le  C(\ov v,\ov\om)\Vert Z(t)\Vert_{-\beta}\int_{t}^{t_1}  \Vert Z(s)\Vert_{-\beta}^2 ds, \quad  0\le t\le t_1< t_*'.
\ee
Similarly,  
\be\la{N1est}
\Vert\check N^+(t,t_1)\Vert_{\beta}\le C(\ov v,\ov\om) \Vert Z(t)\Vert_{-\beta}^2, \quad  0\le t\le t_1<t_*'.
\ee
\subsection{Decay of linearized dynamics}\label{LinDynDecay}
Note that even for the ``frozen'' linear equation, dispersive decay for all solutions does not hold without an orthogonality condition of type \eqref{orth}.
Namely, by \eqref{Atanformv}--\eqref{Atanformv1}, equation \eqref{line} with $u=v$ admits the secular solutions
\be\la{sec-sol}
X(t)=\sum\limits_{j=1,2,5}C_j\tau_j(v,\om)+\sum\limits_{j=1,2}D_j\big[t\tau_j(v,\om)+\tau_{j+2}(v,\om)\big]+D_5\big[t\tau_5(v,\om)+\tau_{6}(v,\om)\big].
\ee
Therefore, we will apply the corresponding symplectically orthogonal projection which eliminates  the ``runaway solutions'' \eqref{sec-sol}.
\begin{defin}
i)  For $v\in {\cal V}$ and ${\omega}\in\R$ denote by ${\bf\Pi}_{v,\omega}$ the symplectically orthogonal projection
of ${\cal E}$ onto the tangent space ${\cal T}_{S(\si)}{\cal S}$, and  ${\bf P}_{v,\omega}:={\bf I}-{\bf\Pi}_{v,\omega}$. 
\\
ii) Denote by ${\cal Z}_{v,\omega}={\bf P}_{v,\omega}{\cal E}$ the space symplectically orthogonal to ${\cal T}_{S(\si)}{\cal S}$.
\end{defin}
\begin{remark}\la{rPiv}
Note that, by  linearity, 
\be\la{Piv}
{\bf\Pi}_{v,\omega}Z=\sum{\bf\Pi}_{jl}(v,{\omega})\tau_j(v,{\omega})\Om(\tau_l(v,{\omega}),Z),\quad Z\in{\cal E},
\ee
with some smooth coefficients ${\bf\Pi}_{jl}(v,\omega)$. 
\end{remark}
Now  the symplectic orthogonality  \eqref{orth}  can be written in the following equivalent forms:
\be\la{PZ}
{\bf\Pi}_{v(t),{\omega}(t)} Z(t)=0,~~~~{\bf P}_{v(t),{\omega}(t)}Z(t)= Z(t),~~~~~~~~~0\le t<t_*'.
\ee
\begin{remark}\la{rZ}
{\rm i)
The tangent space ${\cal T}_{S(\si)}{\cal S}$ is invariant under the operator $L_{v,\omega}$ by Lemma \ref{ceig2}.  Hence,
the space  ${\cal Z}_{v,\omega}$ is also invariant: $L_{v,\omega}Z\in {\cal Z}_{v,\omega}$
for {\it sufficiently smooth}  $Z\in {\cal Z}_{v,\omega}$.\\
ii) The condition  $\Om(Z,\tau_5(v,\omega))=0$ and definition \eqref{inb} of $\tau_5$ imply that the  last component  $\chi$ of  $Z\in {\cal Z}_{v,\omega}$ vanishes.
}
\end{remark}
In Section \ref{X-decay} below we will  prove the following theorem.
\begin{theorem}\la{lindecay}
 Let all conditions of Theorem \ref{main} hold, $|v|\le\ov v<1$, $|\omega|\le \ov\om$, and $X_0\in{\cal Z}_{v,{\omega}}\cap\cE_\beta$ with $\beta>5/2$.
Then 
 $X(t)=e^{L_{v,\omega}t}X_0\in C(\R, {\cal Z}_{v,{\omega}}\cap {\cal E}_{-\beta})$, and the following decay holds:
\be\la{frozenest}
\Vert e^{L_{v,\omega}t}X_0\Vert_{-\beta}\le
C_{\beta}(\ov v,\ov\om)(1+|t|)^{-3/2}\Vert X_0\Vert_{\beta},~~~~~~~~
\,\,\,t\in\R.
\ee
\end{theorem}
\subsection{Decay of transversal component}\label{TransDecay}
In  Section \ref{sol-as}  we will derive  our main Theorem \ref{main}  from the following time decay of the transversal component $Z(t)$:
\begin{pro}\la{pdec}
 Let all conditions of Theorem \ref{main} hold. Then $t_*=\infty$, and
\be\la{Zdec}
\Vert Z(t)\Vert_{-\beta}\le {C(\ov v,\ov\om,d_\beta)}(1+t)^{-3/2},\qquad t\ge0,\quad\beta>5/2.
\ee
\end{pro}
\begin{proof}
{\it Step i)} Equation \eqref{redy1} implies
\[
{\bf P}_1 \check Z(t)=e^{L_1t}{\bf P}_1\check Z(0)+\int_0^te^{L_1(t-s)}{\bf P}_1\big(\mathfrak{L}(s,t_1)\check Z(s)+\check N^+(s,t_1)\big)ds.
\]
We have used here that the projection ${\bf P}_1:={\bf P}_{v(t_1),{\omega}(t_1)}$ commutes with the group $e^{L_1t}$, since the space 
${\cal Z}_1:={\bf P}_1{\cal E}$ is invariant with respect to $e^{L_1t}$ by Remark \ref{rZ}.
Applying \eqref{frozenest}, we obtain
\be\la{bPZ}
\Vert {\bf P}_1 \check Z(t)\Vert_{-\beta}
\le C\Big(\frac{\Vert {\bf P}_1 \check Z(0)\Vert_{\beta}}{(1+t)^{3/2}}
+\int_0^t\frac{\Vert {\bf P}_1[\mathfrak{L}(s,t_1)\check Z(s)+\check N^+(s,t_1)]\Vert_{\beta}}{(1+|t-s|)^{3/2}}ds\Big).
\ee
The operator ${\bf P}_1={\bf I}-{\bf\Pi}_1$ is continuous in ${\cal E}_\beta$ by \eqref{Piv}.
Hence, \eqref{B1Z1est},  \eqref{N1est}, and \eqref{bPZ}  imply that 
\be\la{duhest}
\Vert {\bf P}_1 \check Z(t)\Vert_{-\beta}
\le C\Big(\frac{\Vert Z(0)\Vert_{\beta}}{(1+t)^{3/2}}
+\int_0^t\frac{\Vert Z(s)\Vert_{-\beta}}{(1+|t-s|)^{3/2}}\big[\Vert Z(s)\Vert_{-\beta}+\int_s^{t_1}\Vert Z(\tau)\Vert_{-\beta}^2 d\tau\big]ds\Big)
\ee
 for   $t_1<t_*'$  and $0\le t\le t_1$. Finally,  we are going to  replace ${\bf P}_1 \check Z(t)$ by $Z(t)$ in the left-hand side of \eqref{duhest}.
For  this we introduce the ``majorant'' 
\be\la{maj1}
M(t):=
\sup_{s\in[0,t]}(1+s)^{3/2}\Vert Z(s)\Vert_{-\beta}
\ee
and  further reduce the exit time: 
 \[
t_*''=\sup \{t\in[0,t_*'):\ M(s)\le \ve,~~0\le s\le t\}.
\]
Here $\ve$ is a fixed positive number, which we will specify below.
\begin{lemma}\la{Z1P1Z1} 
For sufficiently small $\ve>0$ and  $t_1<t_*''$,
the following bound holds:
$$
\Vert Z(t)\Vert_{-\beta}\le C(\ov v,\ov\omega)\Vert {\bf P}_1 \check Z(t)\Vert_{-\beta},\quad0\le t \le t_1. 
$$
\end{lemma}
The proof is based on the symplectic orthogonality \eqref{PZ} and is similar to that of Lemma 9.2 in \cite{IKV2012}.
Lemma \ref{Z1P1Z1} together with \eqref{duhest}  implies
\[
\Vert Z(t)\Vert_{-\beta}
\le C\Big(\frac{\Vert Z(0)\Vert_{\beta}}{(1+t)^{3/2}}+\int\limits_0^t\frac{\Vert Z(s)\Vert_{-\beta}}{(1+t-\!s)^{3/2}}
\Big[\int\limits_s^{t_1}\Vert Z(\tau)\Vert_{-\beta}^2 d\tau+\Vert Z(s)\Vert_{-\beta}\Big]ds\Big),~~t\le t_1<t_*''.
\]
Multiplying both sides  by $(1+t)^{3/2}$ and taking the supremum in $t\in[0,t_1]$, we get:
\[
M(t_1) \le C\Big(\Vert Z(0)\Vert_{\beta}+\!
\sup_{t\in[0,t_1]}\int\limits_0^t \frac{(1+t)^{3/2}}{(1+t-\!s)^{3/2}}\Big[\frac{M^2(s)}{(1+s)^{3}}+\frac{M(s)}{(1+s)^{3/2}}
\int\limits_s^{t_1}\frac{M^2(\tau)d\tau}{(1+\tau)^{3}}\Big]ds\Big),~~ t_1<t_*''.
\]
Taking into account that $M(t)$ is a monotonically increasing function, we obtain that
\[
M(t_1)\le C\Big(\Vert Z(0)\Vert_{\beta}+\big(M^3(t_1)+M^2(t_1)\Big),\quad t_1<t_*''.
\]
This inequality implies that $M(t_1)$ is bounded for $t_1<t_*''$,
and moreover
\be\la{m2est}
M(t_1)\le C_1\Vert Z(0)\Vert_{\beta},~~~~~~~~~t_1<t_*''\,,
\ee
since $M(0)$ is sufficiently small by \eqref{close}.
\\
{\it Step ii)} The constant $C_1$ in  \eqref{m2est} does not depend on
$t_*$, $t_*'$, and $t_*''$ by Lemma \ref{Z1P1Z1}.
We choose $d_{\beta}$ in \eqref{close} so small that
$\Vert Z(0)\Vert_{\beta}<\ve/(2C_1)$.
Then  \eqref{m2est} implies that $t''_*=t'_*$, and therefore  \eqref{m2est} holds for all $t_1<t'_*$.
Further,   
$u(s)-v(t_{1})=\dot c(s)+\ds\int_s^{t_1}\dot v(\tau)d\tau$ by \eqref{vw}.
Hence, Lemma \ref{mod} and definition \eqref{dd1} imply that
\beqn\nonumber
 |\ell_{1}(t)|&\le& \int_{t}^{t_{1}}\left( |\dot{c}(s)|+\int_{s}^{t_{1}}|\dot{v}(\tau )|d\tau\right)ds\\
\nonumber
& \le& CM^2(t_1)\int_{t}^{t_{1}}\left(\frac{1}{(1+s)^3}+\int_{s}^{t_{1}}\frac{d\tau}{(1+\tau)^3}\right)ds\le CM^2(t_1)\le C_2\ve^2,\quad t<t_*'.
\eeqn 
  We choose $\ve$ so
small that $C_2\ve^2\le 1$. Then $t'_*=t_*$.
Therefore, \eqref{m2est} holds for all $t_1<t_*$.  Hence,
$\Vert Z(t)\Vert_{-\beta}<z_{-\beta}(\ov v,\ov\om)/2$ for $0\le t < t_*$
if $\Vert Z(0)\Vert_{\beta}$ is sufficiently small.  Finally, this implies that $t_*=\infty$, hence also $t''_*=t'_*=\infty$ and
\eqref{m2est} holds for all $t_1>0$ if $d_{\beta}$ is small enough. 
\end{proof}
\subsection{Proof of  Theorem \ref{main}}\label{sol-as}
{\it Step i)}  First, we obtain asymptotics \eqref{qq-as}--\eqref{wphi-as}  for the vector components.\\
From  \eqref{parameq} and \eqref{Zdec} it follows that
\[
|\dot c(t)|+|\dot v(t)|+|\dot\omega(t)|+|\dot\vartheta(t)|={\cal O}(t^{-3}), \quad  t\to\infty.
\]
Hence,
\be\la{cv-as} 
c(t)=c_+ +{\cal O}(t^{-2}),~~ v(t)=v_+ +{\cal O}(t^{-2}),~~ \omega(t)=\omega_{+} +{\cal O}(t^{-2}),~~
\vartheta(t)=\vartheta_+ +{\cal O}(t^{-2}),~~  t\to\infty.
\ee 
Now \eqref{solvom}, \eqref{dot-q}, \eqref{JAro}, \eqref{Zdec},  and \eqref{cv-as} imply
\be\la{dq}
\dot q(t)=v(t)+\frac 1m\big[\pi(t)-\langle\Lam(t),\rho\rangle-\omega(t)JPr(t)\big]+{\cal O}(\Vert Z(t)\Vert_{-\beta}^2)=v_++{\cal O}(t^{-3/2}),\quad  t\to\infty.
\ee
Moreover, \eqref{dec},  \eqref{vw}, \eqref{Zdec}, and  \eqref{cv-as}  give
\[
q(t)=b(t)+r(t)=c(t)+\int_0^tv(s)ds+{\cal O}(t^{-3/2})=a_++v_+t+{\cal O}(t^{-1}),\quad  t\to\infty.
\]
Hence,  \eqref{qq-as}  follows. 
Further,  \eqref{mls3}, \eqref{solvom}, \eqref{vw}, \eqref{Zdec}, and \eqref{cv-as}  imply that
\beqn\nonumber
I\dot\varphi(t)&\!\!\!=\!\!\!&M(t)-\langle A(x,t),J\varrho(x\!-\!q(t))\rangle
=M_{v(t),\omega(t)}+\chi(t)-\langle A_{v(t),\omega(t)}(y)+\Lam(y,t),J\varrho(y\!-\!r(t))\rangle\\
\nonumber
&\!\!\!=\!\!\!&I\omega(t)-\langle A_{v(t),\omega(t)},J(\varrho(y\!-\!r(t))-\!\varrho(y))\rangle
+ \chi(t)-\langle\Lambda(y,t),J\varrho(y\!-\!r(t))\rangle\\
\la{ww}
&\!\!\!=\!\!\!&I\omega_{+}+ {\cal O}(t^{-3/2}). 
\eeqn
Finally,   \eqref{dec}, \eqref{vw}, \eqref{Zdec}, and  \eqref{cv-as}  imply
$$
\varphi(t)=\gamma(t)+\psi(t)=\vartheta(t)+\int\limits_0^t {\omega}(s)ds+{\cal O}(t^{-3/2})=\theta_++{\omega}_+t+ {\cal O}(t^{-1}),\quad  t\to\infty.
$$
Hence, \eqref{wphi-as} follows.
\\
{\it Step ii)} Now we obtain asymptotics \eqref{AP-as}--\eqref{r-as}  for the fields.\\
For the field part $F(t)=(A(t), \Pi(t))$ of the solution,  we  define the accompanying soliton field as 
$F_{{\rm v}(t),{\rm w}(t)}(t)\! = (A_{{\rm v}(t),{\rm w}(t)}(x\!-\!q(t)), \Pi_{{\rm v}(t),{\rm w}(t)}(x\!-\!q(t)))$, where  ${\rm v} = \dot q$, ${\rm w}=\dot\varphi$. 
Then for the difference 
${\cal D}(t)=(\Lam(t),\Psi(t)) = F(t) -F_{{\rm v}(t),{\rm w}(t)}$,  the first two equations of  \eqref{mls3} imply 
\[
\dot {\cal D}(x,t)=\!\begin{pmatrix} 0& 1\\
\Delta&0
\end{pmatrix}{\cal D}(x,t)-R(t), ~~{\rm where}~~R(t):=(\dot {\rm v}(t)\cdot\!\na_{{\rm v}}+\dot{\rm w}(t)\na_{{\rm w}})F_{{\rm v}(t),{\rm w}(t)}(x-q(t)).
\]
One has 
\be\la{d-1}
{\cal D}(t)=W_0(t){\cal D}(0)-\int_0^t W_0(t-s)R(s)ds=W_0(t)\Phi_++r_+(t),
\ee
where we denote
$\Phi_+=\ds {\cal D}(0)-\int_0^{\infty}W_0(-s)R(s)ds$ and  $r_+(t):=\ds\int_t^{\infty}W_0(t-s)R(s)ds$.\\
From  \eqref{mls3}, \eqref{Hs5}, \eqref{Zdec}, and \eqref{dq}--\eqref{ww}  it follows that
\beqn\nonumber
m\dot {\rm v}(t)&=&-\langle \Pi(x,t),\rho(x\!-\!q(t))\rangle
+{\rm w}(t) \langle \na\cdot\!J A(x,t),\varrho(x\!-\!q(t))\rangle
+J{\rm v}(t)\langle \na\cdot \!J A(x,t),\rho(x\!-\!q(t))\rangle\\
\nonumber
&=&-\langle \Psi(y,t),\rho(y\!-\!r(t))\rangle+\langle ({\rm v}\cdot\na) A_{{\rm v},{\rm w}}(y),\rho(y\!-\!r(t))-\rho(y)\rangle\\
\nonumber
&+&{\rm w}(t)\langle \na\cdot J \Lam(y,t),\varrho(y\!-\!r(t))\rangle+{\rm w}(t)\langle \na\cdot J A_{{\rm v},{\rm w}}(y),\varrho(y\!-\!r(t))-\varrho(y)\rangle\\
\nonumber
&+&J{\rm v}(t)\langle \na\cdot J \Lam(y,t),\rho(y\!-\!r(t))\rangle
+J{\rm v}(t)\langle \na\cdot J A_{{\rm v},{\rm w}}(y),\rho(y\!-\!r(t))-\rho(y)\rangle={\cal O}(t^{-3/2})
\eeqn
as  $t\to\infty$, since by \eqref{solit3}, 
\beqn\nonumber
&&\langle ({\rm v}\cdot\na) A_{{\rm v},{\rm w}},\rho\rangle+{\rm w} \langle \na\cdot JA_{{\rm v},{\rm w}},  \varrho\rangle
+J{\rm v}\langle \na\cdot J A_{{\rm v},{\rm w}},\rho\rangle\\
\nonumber
&&
={\rm w}\Big[\langle  \frac{({\rm v}\cdot k)  J\na \hat\rho}{\hat D_0},\hat\rho\rangle
+\langle \frac{(k\cdot J{\rm v})\hat \rho(k)}{\hat D_0},  \na\hat\rho\rangle
-J{\rm v}\langle \frac{k\cdot \na\hat\rho}{\hat D_0},\hat\rho\rangle\Big]
={\rm w}{\rm v}\langle\frac{(k_1\na_2-k_2\na_1)\hat\rho}{\hat D_0},\hat\rho\rangle=0.
\eeqn
Similarly,  \eqref{mls3},  \eqref{Hs5}, \eqref{Zdec}, and \eqref{dq}--\eqref{ww}    imply 
\beqn\nonumber
I\dot{\rm w}(t)&\!\!\!=\!\!\!&\langle \Pi(x,t), J\varrho(x-q(t))\rangle+\langle  A(x,t),({\rm v}\cdot\na)J\varrho(x-q(t))\rangle\\
\nonumber
&\!\!\!=\!\!\!&-\langle ({\rm v}\cdot\na) A_{{\rm v},{\rm w}}(y),J\big[\varrho(y\!-\!r(t))-\varrho(y)\big]\rangle
+\langle  A_{{\rm v},{\rm w}}(y),({\rm v}\cdot\na)J\big[\varrho(y\!-\!r(t))-\varrho(y)\big]\rangle\\
\nonumber
&\!\!\!+\!\!\!&\langle \Psi(x,t),J\varrho(x\!-\!q(t))\rangle+\langle \Lam(x,t),({\rm v}\cdot\na)J\varrho(x\!-\!q(t))\rangle={\cal O}(t^{-3/2}),~~ t\to\infty.
\eeqn
Hence,   $\Vert R(t)\Vert_{\cal F} = {\cal O}(t^{-3/2})$.
Therefore, $\Phi_{+}\in{\cal F}$ and  $\Vert r_{+}(t)\Vert_{\cal F}= {\cal O}(t^{-1/2})$ due to the unitarity of $W_0(t)$. 
Finally, asymptotics \eqref{AP-as}--\eqref{r-as} follow from \eqref{d-1}.
\setcounter{equation}{0}
\section{Dispersive decay}\label{lin-dyn}
In the remaining sections we prove the dispersive decay for solution to the  linear equation \eqref{line}  with  $u=v$,  
as stated  in Theorem \ref{lindecay}.
The structure of the system \eqref{line}  allows us to obtain an equation  for $r(t)$, which does not contain other variables.  
Using the Fourier--Laplace  transform, we  solve this equation and prove the decay \eqref{frozenest} for $r(t)$. The decay for the remaining vector components 
 follows from their  expressions in terms of $r(t)$. To prove the decay of  fields, we apply the Duhamel representation and 
the  known dispersive decay for the 2D wave equation obtained in  \cite{K2010, KK2023}.
\subsection{Reduced equation}\label{redu}
First, we change variables in  equation  \eqref{line}  to simplify its structure. Denote 
\be\la{nu-def}
\phi=\frac{1}{m}\big(\pi-\langle \Lam,\rho\rangle\big),\qquad \nu=\frac{1}{I}\big(\langle \Lam, J\varrho\rangle-r\cdot PJv\big).
\ee
If we prove  decay for $\phi$ and $\Lam$, then $\pi$  has the corresponding decay as well.
By \eqref{AA} and \eqref{B-op},
\[
\dot\phi=\frac{1}{m}\big[\dot\pi-\langle \dot\Lam,\rho\rangle\big]
=\frac{1}{m}\Big[-\langle v\cdot \Lam,\na\rho\rangle+{\omega}\langle \na\cdot J\Lam,\varrho\rangle-{\omega} PJ\phi
+\nu PJv -Gr-\langle \Psi+v\cdot\na\Lam,\rho\rangle\Big].
\]
Recall that  $\chi=0$ (see Remark \ref{rZ}).  Moreover, the RHS of \eqref{line} does not depend on $\psi$ by  \eqref{AA}.
Thus, it  is natural to consider the reduced equation  for
${\cal X}=(\Lambda, \Psi, r, \phi)$.  Namely, equation  \eqref{line} implies
\be\la{bfA}
\dot{\cal X}\!={\cal L}{\cal X}=\begin{pmatrix}
\Psi+v\cdot\na\Lambda \\
\De\Lambda+v\cdot\na\Psi+ {\cal P}[\rho\phi - \frac {\omega}m\rho JPr-\!v(r\cdot\na\rho)]+\big({\omega}(r\cdot \na)-\nu\big)J\varrho \\
\phi-\frac{\omega}m JPr\\
-\frac{1}m\Big(\langle v\cdot \Lam,\na\rho\rangle-{\omega}\langle \na\cdot \!J\Lam,\varrho\rangle+{\omega} PJ \phi-\nu PJv+Gr
+\langle \Psi+v\cdot\na\Lam,\rho\rangle\Big)
\end{pmatrix}.
\ee
Applying the Fourier--Laplace transform
$\ti {\cal X}(\lam):=\int_0^\infty e^{-\lam t}{\cal X}(t)dt$, $\rRe\lam>0$, we obtain
\be\la{eq-main}
\left\{\!
\ba{l}
\ti\Psi+v\cdot\na\ti\Lambda -\lam\ti\Lambda=-\Lambda_0\\
\De\ti\Lambda+v\cdot\!\na\ti\Psi+{\cal P}\big[\big(\lam\ti r-v(\ti r\cdot\na)\big)\rho\big]
+\big( {\omega}(\ti r\cdot\na)\!-\ti\nu\big)J\varrho-\lam\ti\Psi=-\Psi_0+{\cal P}[r_0\rho]\\
\ti\phi-\frac{\omega}m JP\ti r-\lam \ti r=-r_0\\
\langle v\cdot \ti\Lam,\na\rho\rangle+({\omega} PJ+m\lam)(\lam+\frac{\omega}m JP)\ti r-{\omega}\langle \na\cdot \!J\ti\Lam,\varrho\rangle
+G\ti r-\ti\nu PJv+\lam\langle \ti\Lambda,\rho\rangle\\
\quad \quad =\pi_0+({\omega} PJ\!+m\lam)r_0
\ea\right.
\ee
In Fourier space, the first two equations  become
\[
\left\{
\ba{rcl}
-(i(v\cdot k)+\lam)\hat{\ti\Lam}(\lam,k)+ \hat{\ti\Psi}(\lam,k)&=&-\hat\Lam_0(k)
\\
-k^2\hat{\ti \Lam}(\lam,k)-(i(v\cdot k)+\lam)\hat{\ti\Psi}(\lam,k)&=&-\hat\Psi_0(k)+\widehat{{\cal P}r_0\rho}(k)-\Phi(\lam,k),
\ea\right.
\]
where 
\be\la{K}
\Phi(\lam,k):=\lam\widehat{{\cal P}\ti r\rho}(k)+i(k\cdot\ti r)\widehat{{\cal P} v\rho}(k)+\big(i\ti \nu-{\omega} (k\cdot \ti r)\big)J\na\hat\rho.
\ee
Hence,
\be\la{hat-Lam}
\hat{\ti \Lam}(\lam,k)=\frac{\Phi_0(\lam,k)}{\hat D(\lam,k)}+\frac{\Phi(\lam,k)}{\hat D(\lam,k)}, \qquad\rRe\lam>0,
\ee
where
\be\la{KK}
\hat D(\lam,k):=(i(v\cdot k)+\lam)^2+k^2,\qquad
\Phi_0(\lam,k):=(i(v\cdot k)+\lam)\hat\Lam_0(k)+\hat\Psi_0(k)-\widehat{{\cal P}r_0\rho}(k).
\ee
Substituting \eqref{hat-Lam} into the fourth equation of \eqref{eq-main}, we get
\be\la{eq-Pi}
i\langle \frac{v\cdot\! \Phi}{\hat D(\lam)},k\hat\rho\rangle-{\omega}\langle \frac{k\cdot\! J \Phi}{\hat D(\lam)},\na\hat\rho\rangle
+\big(Q+{\omega}^2 F\!+{\omega}\lam(PJ\!+\!JP)+m\lam^2\big)\ti r-\ti\nu PJv+\lam\langle \frac{\Phi}{\hat D(\lam)},\hat\rho\rangle\!=K_0,
\ee
where
\be\la{F1}
K_0=K_0(\lam)=\pi_0+({\omega} PJ+m\lam)r_0-i\langle\frac{ v\cdot \Phi_{0}}{\hat D(\lam)},k\hat\rho\rangle
+{\omega}\langle\frac{ k\cdot J\Phi_0}{\hat D(\lam)},\na\hat\rho\rangle-\lam\langle \frac{\Phi_0}{\hat D(\lam)},\hat\rho\rangle.
\ee
Let $F(\lam)$, $V(\lam)$, $N(\lam)$,  $P(\lam)$, $Q(\lam)$,  and $U(\lam)$ be  matrices with  entries 
\beqn\nonumber
\!\!\!\!\!\!\!\!\!\!&&\!\!\!\!\! F_{jl}(\lam)=\!\int\! \frac{k_jk_l|\na\hat\rho|^2}{\hat D(\lam)} dk, ~~ V_{jl}(\lam)=\!\int\!\frac{ (v\cdot\! Jk)|\hat\rho|^2k_j k_l}{k^2\hat D(\lam)} dk,\quad
N_{jl}(\lam)=\!\int \!\frac{(v\cdot \!Jk)\hat\rho k_j\na_l\hat\rho}{\hat D(\lam)} dk,\\ 
\la {PQF}
\!\!\!\!\!\!\!\!\!\!&&\!\!\!\!\! P_{jl}(\lam)= \!\int\! \frac{k_l\na_j\hat\rho(k)}{\hat D(\lam)}\hat\rho(k) dk,\quad
Q_{jl}(\lam)=\!\int \!\big[v^2k^2\!-\!(v\cdot k)^2\big]\frac{k_lk_j|\hat\rho(k)|^2dk}{k^2\hat D(\lam)}, ~~j,l=1,2.\\
\la{UU}
\!\!\!\!\!\!\!\!\!\!&&\!\!\!\!\! U_{12}(\lam)=U_{21}(\lam)=-\int\frac{k_1k_2|\hat\rho|^2dk}{k^2\hat D(\lam)}, \quad
U_{jj}(\lam)=\int\frac{(k^2-k_j^2)|\hat\rho|^2dk}{k^2\hat D(\lam)},~~ j=1,2.
\eeqn
Note that $P(0)=P$, $Q(0)=Q$, $F(0)=F$, where $P,Q,F$ are given by \eqref{c-ro}. 
We also define the functions
\be\la{fh-def}
f(\lam)=\int\frac{\hat\rho }{\hat D(\lam)}J\na\hat\rho\, dk, \quad
h(\lam)=\int\frac{|\na\hat\rho|^2}{\hat D(\lam)}k\, dk,\quad \varkappa(\lam)=\int\frac{|\na\hat\rho|^2}{\hat D(\lam)}dk.
\ee
Then \eqref{Pi-e} and \eqref{K} imply
\beqn\la{Pi-1}
\langle\frac {\Phi}{\hat D(\lam)},\hat\rho\rangle
&\!\!\!=\!\!\!&\big(\lam U(\lam)-{\omega} JP(\lam) +i JV(\lam)\big)\ti r+i\ti \nu f(\lam),\\
\label{Pi-2}
\langle\frac{ k\cdot\! J\Phi}{\hat D(\lam)},\na\hat\rho\rangle
&\!\!\!=\!\!\!&\big(\lam P(\lam)J+{\omega} F(\lam)-iN(\lam)\big)\ti r-i\ti \nu h(\lam),\\
\label{Pi-3}
\langle \frac {v\cdot\! \Phi}{\hat D(\lam)},k\hat\rho\rangle
&\!\!\!=\!\!\!&\big(iQ(\lam)-{\omega} N(\lam)-\lambda V(\lam)J\big)\ti r-i\ti\nu P(\lam)Jv,\\
\la{Pi-4}
\langle \frac {\Phi}{\hat D(\lam)}, J\hat\varrho\rangle
&\!\!\!=\!\!\!&\ti r\cdot \big(P(\lam)Jv+i\lam f(\lam)-i{\omega} h(\lam)\big)-\ti\nu\varkappa(\lam).
\eeqn
Denote $\breve P(\lam):=P-P(\lam)$, $\breve Q(\lam):=Q-Q(\lam)$, $\breve F(\lam):=F-F(\lam)$.   Using  \eqref{nu-def},  \eqref{hat-Lam}, and  \eqref{Pi-1}, we get
\be\la{nu-rep}
\ti\nu=\frac{1}{\kappa(\lam)}\Big[\ti r\cdot \big(i\lam f(\lam)-i{\omega} h(\lam)-\breve P(\lam)Jv\big)+\langle \frac{\Phi_0}{\hat D(\lam)},J\hat\varrho\rangle\Big],
~~~{\rm where} ~~\kappa(\lam):=I+\varkappa(\lam).
\ee
Evidently, $\kappa(0)>0$. Moreover, $\kappa(\lam)\ne 0$ for $\rRe\lam>0$.
Substituting \eqref{Pi-1}--\eqref{Pi-4} and \eqref{nu-rep}  into the LHS of \eqref{eq-Pi}, we obtain
\beqn \nonumber
&&\Big[\breve Q+{\omega}^2\breve F+\lam^2 U(\lam)+{\omega}\lam\big(\breve PJ\!+\!J\breve P\big)+i\lam(JV(\lam)-V(\lam)J)+m\lam^2\Big]\ti r\\
\la{M13}
&&+\ti r\cdot\frac{i\lam f(\lam)-i{\omega} h(\lam)-\breve PJv}{\kappa(\lam)}\Big[i\lam f(\lam)+i{\omega} h(\lam)-\breve PJv\Big]={\cal K}_0(\lam), 
\eeqn
where
\be\la{F1+}
{\cal K}_0(\lam):=K_0(\lam)-\frac{1}{\kappa(\lam)}\langle \frac{\Phi_0}{\hat D(\lam)},J\hat\varrho\rangle\Big[i\lam f(\lam)+i{\omega} h(\lam)-\breve PJv\Big].
\ee
We rewrite \eqref{M13}  as 
\be\la{M-def}
{\cal M}(\lam)\ti r(\lam)={\cal K}_0(\lam),~~ {\rm where}
\quad {\cal M}(\lam)=\begin{pmatrix}a_{11}(\lam)+b_{11}(\lam)&a_{12}(\lam)+b_{12}(\lam)\\
a_{21}(\lam)+b_{21}(\lam)&a_{22}(\lam)+b_{22}(\lam)
\end{pmatrix},
\ee
where 
\beqn\nonumber
a_{jl}&=&\breve Q_{jl}+{\omega}^2\breve F_{jl}+\lam^2(m\delta_{jl}+U_{jl})+{\omega}\lam\big(\breve PJ\!+\!J\breve P\big)_{jl}+i\lam(JV-VJ)_{jl},\\
\label{aabb}
b_{jl}&=&\frac{1}{\kappa}\big(i\lam f_j-i{\omega} h_j-(\breve PJv)_j\big)\big(i\lam f_l+i{\omega} h_l-(\breve PJv)_l\big).
\eeqn
The invertibility of the matrix ${\cal M}(\lam)$ for $\rRe\lam>0$ follows from the next lemma. 
\begin{lemma}\la{L-KK} 
The operator ${\cal L}-\lam:{\cal F}\oplus \R^4\to{\cal F}\oplus \R^4$ has a bounded inverse  for $\rRe\lam>0$.
\end{lemma}
The proof is similar to that of  Lemma 14.1 in \cite{KK2023}. By \eqref{M-def},
\be\la{detM}
{\cal M}^{-1}=\frac{1}{{\rm det}\, {\cal M}}\begin{pmatrix} 
a_{22} +b_{22}& -(a_{12}+b_{12})\\
a_{12}+b_{12}&  a_{11} +b_{11}
\end{pmatrix},\quad  {\rm det}\, {\cal M}=(a_{11} +b_{11})(a_{22} +b_{22})+(a_{12}+b_{12})^2.
\ee
Lemma \ref{L-KK} together with  \eqref{M-def} implies
\be\la{r-sol}
\ti r(\lam)={\cal M}^{-1}(\lam){\cal K}_0(\lam),\quad \rRe\lam>0. 
\ee
\subsection{Regularity  on the imaginary axis}\label{det-as}
First, let us show that the limit matrix ${\cal M}(i\mu)={\cal M}(i\mu+0)$ exists for $\mu\in\R$.
For this, we return to the coordinate representation. In the case $v=(|v|,0)$, formula \eqref{KK} implies
\[
\hat D(\lam,k)=k^2+(i|v|k_1+\lam)^2=\frac{1}{\ga^2}(k_1+i|v|\lam\gamma^2)^2+k_2^2+\ga^2\lam^2,\qquad \ga:=1/\sqrt{1-v^2}.
\]
Hence, for $\rRe\lam>0$, we obtain
\be\la{g-rep}
{\cal R}(\lam,z):=\frac{1}{2\pi}\int \frac{e^{-ikz}dk}{\hat D(\lam,k)}
=\gamma e^{-|v|\lam\gamma \tilde z_1}\int \frac{e^{-ik\tilde z}dk}{k_1^2+k_2^2+\ga^2\lam^2}=\gamma e^{-|v|\lam\gamma\tilde z_1}R(-\gamma^2\lam^2,\tilde z)
,\quad \tilde z=(\gamma z_1,z_2).
\ee
where $R(\zeta, x-y)$ is the integral kernel of the resolvent $R(\zeta)=(-\Delta-\zeta)^{-1}$ of the 2D Laplacian.
Let ${\cal R}(\lam)$ be the integral operator with kernel ${\cal R}(\lam,x-y)$. Then
the well-known properties of $R(\zeta)$ (see, for example, \cite{A, KK2012}) imply the following.\\
1) In $H_{\al}^0\to H_{-\al}^{2}$ with $\al>1/2$,
the convergence holds:
\be\la{LAP}
{\cal R}(i\mu+\ve)\to {\cal R}(i\mu+ 0),\quad \ve\to +0,
\ee
where  ${\cal R}(i\mu+0)$ is the integral operator with kernel
\be\la{g-rep1} 
{\cal R}(i\mu+0,z)=\gamma e^{-i|v|\mu\gamma\tilde z_1}R(\gamma^2\mu^2\mp 0,\ti z)
=\gamma e^{-i|v|\mu\gamma\tilde z_1}R_{\mp}(\gamma^2\mu^2,\ti z),\quad  \pm\mu>0.
\ee
2) The following asymptotics hold as $\mu\to\infty$:
\be\la{R-as}
\Vert {\cal R}^{(k)}(i\mu+0)\Vert_{H^0_\al \to H^{l}_{-\al}}
 ={\cal O}(|\mu|^{-1+l}),\quad k=0,1,\dots, \quad \al>1/2+k, \quad l=-1,0,1,2.
 \ee
Asymptotics \eqref{R-as}  also hold in the case $k=1,2,\dots$ and  $l=-2$: 
\begin{lemma}\la{as-new}
\be\la{R-as1}
\Vert {\cal R}^{(k)}(i\mu+0)\Vert_{H^0_\al \to H^{-2}_{-\al}}
={\cal O}(|\mu|^{-3}),\quad \mu\to\infty, \quad k=1,2, \dots, \quad \al>1/2+k.
\ee
\end{lemma}
\begin{proof}
It suffices to prove  that for $R_{\pm}=R(\zeta\pm i0)$ the following asymptotics hold:
\be\la{R-as2}
\Vert R^{(k)}_{\pm}(\zeta)\Vert_{H^0_\al \to H^{-2}_{-\al}}
={\cal O}(\zeta^{-\frac{3+k}2}),\quad \zeta\to+\infty, \quad k=1,2,\dots, \quad \al>1/2+k.
\ee
These  asymptotics  follow from the well-known  asymptotics 
\be\la{R-as3}
\Vert R^{(k)}_{\pm}(\zeta)\Vert_{H^0_\al \to H^{0}_{-\al}}
={\cal O}(\zeta^{-\frac{1+k}2}),\quad \zeta\to +\infty, \quad\zeta>0,\quad k=0, 1,2,\dots, \quad \al>1/2+k.
\ee
Indeed, the relation $(1-\Delta)R(\zeta)=1+(1+\zeta)R(\zeta)$ implies
$$
R(\zeta)=\frac{-1+(1-\Delta)R(\zeta)}{1+\zeta},\quad  
R'(\zeta)=\frac{1-(1-\Delta)R(\zeta)}{(1+\zeta)^2}+\frac{(1-\Delta)R'(\zeta)}{1+\zeta}.
$$
Hence, \eqref{R-as3} implies
$$
\Vert R'_{\pm}(\zeta)\Vert_{H^0_\al \to H^{-2}_{-\al}}\le \frac{1+\Vert R_{\pm}(\zeta)\Vert_{H^0_\al \to H^{0}_{-\al}}}{(1+\zeta)^2}
+\frac{\Vert R'_{\pm}(\zeta)\Vert_{H^0_\al \to H^{0}_{-\al}}}{(1+\zeta)}={\cal O}(\zeta^{-2}),\quad \zeta\to +\infty, \quad\al>3/2. 
$$
Similarly, for $\al>5/2$,
$$
\Vert R''_{\pm}(\zeta)\Vert_{H^0_\al \to H^{-2}_{-\al}}\le\frac{1+\Vert R_{\pm}(\zeta)\Vert_{H^0_\al \to H^{0}_{-\al}}}{(1+\zeta)^3}
+2\frac{\Vert R'_{\pm}(\zeta)\Vert_{H^0_\al \to H^{0}_{-\al}}}{(1+\zeta)^2}
+\frac{\Vert R''_{\pm}(\zeta)\Vert_{H^0_\al \to H^{0}_{-\al}}}{1+\zeta}={\cal O}(\zeta^{-5/2}) 
$$
as $\zeta\to +\infty$. Therefore, \eqref{R-as2} with $k=1,2$ holds.  For $k=3,4,\dots$ the proof is similar.
\end{proof}
 The convergence \eqref{LAP} and formulas \eqref{M-def}--\eqref{aabb} imply that the matrix ${\cal M}(i\mu)$
exists and its entries are analytic functions of $\mu\in\R$.
Let us show that ${\cal M}(i\mu)$ is invertible for sufficiently large $|\mu|$.
 \begin{lemma}\la{Lem-q-est}
Let condition \eqref{rosym} holds. Then ${\cal M}(i\mu)$ is invertible for sufficiently large $|\mu|$, and 
\be\la{M-as-}
{\cal M}^{-1}(i\mu)=\begin{bmatrix} 
-\frac{1}{m\mu^2}+{\cal O}(|\mu|^{-4})&{\cal O}(|\mu|^{-3})\\
{\cal O}(|\mu|^{-3})&-\frac{1}{m\mu^2}+ {\cal O}(|\mu|^{-4})
\end{bmatrix},\quad \mu\to\infty.
\ee
The asymptotics can be differentiated twice. 
 \end{lemma}
 \begin{proof}
From now on we assume that $v=(|v|,0)$.
 In this case formulas \eqref{PQF}--\eqref{fh-def} imply
\be\la{VP}
V_{11}(\lam)=V_{22}(\lam)=f_1(\lam)=h_2(\lam)=0,\qquad  P_{jl}(\lam)=Q_{jl}(\lam)=F_{jl}(\lam)=U_{jl}(\lam)=0,\quad j\ne l.
\ee
\beqn\nonumber
\!\!\!\!\!\!\!\!\!\!\!\!\!\!\!\!\!\!&&V_{12}(\lam)=\int\frac{|v|k_2^2k_1|\hat\rho|^2dk}{k^2\hat D(\lam)},~~
P_{jj}(\lam)= \!\int\! \frac{k_j\na_j\hat\rho(k)}{\hat D(\lam)}\hat\rho(k) dk,~~F_{jj}(\lam)=\!\int\! \frac{k_j^2|\na\hat\rho|^2dk}{\hat D(\lam)},
\\
\la{PQFSR}
\!\!\!\!\!\!\!\!\!\!\!\!\!\!\!\!\!\!&&Q_{jj}(\lam)=v^2\!\int \!\frac{k_j^2k_2^2|\hat\rho(k)|^2dk}{k^2\hat D(\lam)},\quad
U_{jj}(\lam)=\int\frac{(k^2\!-\!k_j^2) |\hat\rho|^2dk}{k^2\hat D(\lam)},\quad j=1,2.
\eeqn
Hence,
\be\la{PPJ}
 \breve PJ+J\breve P
=\begin{pmatrix} 0 & \breve P_{11} +\breve P_{22} \\ -\breve P_{11} -\breve P_{22}& 0 \end{pmatrix},\quad
 JV-VJ=\begin{pmatrix} 2V_{12} & 0 \\ 0 & -2V_{12} \end{pmatrix}.
\ee
From \eqref{aabb} and \eqref{PPJ}  it follows that
\beqn\nonumber
a_{jj}&=&\breve Q_{jj}+{\omega}^2\breve F_{jj}+\lam^2(m+U_{jj})+2i\lam(-1)^{j+1}V_{12},\qquad a_{12}=-a_{21}={\omega}\lam(  \breve P_{11} +\breve P_{22}),\\
\la{bb}
b_{11}&=&\frac{ {\omega}^2}{\kappa} h_1^2,\quad  b_{22}=\frac{1}{\kappa}\big(i\lam f_2+\breve P_{22}|v|\big)^2,\qquad
b_{12}=-b_{21}=-\frac{i\omega}{\kappa} h_1\big(i\lam f_2+\breve P_{22}|v|\big).
\eeqn
 Applying \eqref{R-as} and \eqref{R-as1}, we obtain  for  $\beta>5/2$ and $k=1,2$, 
 \beqn\nonumber
|\varkappa(i\mu+0)|&\!\!=\!\!&|\langle {\cal R}(i\mu+0)\varrho,\varrho\rangle|
\le C\Vert {\cal R}(i\mu+0)\varrho\Vert_{H^{-1}_{-\beta}}\Vert\varrho\Vert_{H^1_{\beta}}
\le C|\mu|^{-2}\Vert\varrho\Vert_{L^2_{\beta}}\Vert\varrho\Vert_{H^1_{\beta}}\le C(\rho) |\mu|^{-2},\\
\nonumber
|\varkappa^{(k)}(i\mu+0)|&\!\!=\!\!&|\langle {\cal R}^{(k)}(i\mu+0)\varrho,\varrho\rangle|
\le C\Vert {\cal R}^{(k)}(i\mu+0)\varrho\Vert_{H^{-2}_{-\beta}}\Vert\varrho\Vert_{H^2_{\beta}}\\
\la{q-as}
&\!\!\le \!\!&C|\mu|^{-3}\Vert\varrho\Vert_{L^2_{\beta}}\Vert\varrho\Vert_{H^2_{\beta}}\le C(\rho) |\mu|^{-3},\quad  \mu\to\infty.
\eeqn
Similar bounds hold for  $Q_{jj}$, $F_{jj}$,   $P_{jj}$, $U_{jj}$, $V_{12}$, $f_2$,  and $h_1$.
Therefore, \eqref{bb}   implies
$$
a_{jj}(i\mu+0)+b_{jj}(i\mu+0)= -m\mu^{2}+d_{jj}(\mu), \quad 
a_{12}(i\mu+0)+b_{12}(i\mu+0)=i\om\mu(P_{11} +P_{22})+d_{12}(\mu),
$$
where 
$d_{jj}(\mu)= {\cal O}(1)$, $d_{12}(\mu)={\cal O}(|\mu|^{-1})$, $d_{jj}^{(k)}(\mu)= {\cal O}(|\mu|^{-1})$,  and    
$d_{12}^{(k)}(\mu)={\cal O}(|\mu|^{-2})$ as  $\mu\to\infty$  for $k=1,2$. Hence,  \eqref{M-as-} follows. 
\end{proof}
Now we   obtain asymptotics of  ${\cal M}(i\mu)$  as $\mu\to 0$. We will use the following lemma
 \begin{lemma}\la{Dr-as}
 Let $j\in \N$, and let  $p\in C_0^\infty(\R^2)$ satisfy   $\hat p(k)=O(|k|^{j+2})$ as $k\to 0$.  
Then for  any   $\digamma\in L^2_{\al}$ with $\al>j+1/2$ the asymptotics hold 
\beqn\nonumber
\langle \hat\digamma, \frac{\hat p}{\hat D(i\mu+0)}\rangle&=&\langle \hat\digamma ,\frac{\hat p}{\hat D_0}\rangle
+\mu \langle \hat\digamma ,\frac{2|v|k_1\hat p}{\hat D_0^2}\rangle+\mu^2\langle \hat\digamma,\frac{(k^2\!+3v^2k_1^2)\hat p}{\hat D_0^3}\rangle+\dots\\
\la{int-rho-as}
&\dots&+\mu^j\langle \hat\digamma,\frac{P_{j}(k)\hat p}{\hat D_0^{j+1}}\rangle+{\cal O}(|\mu|^{j+1/2}),~~\mu\to 0.
\eeqn
Here $P_j(k)$  is a polynomial of degree $j$. These asymptotics can be differentiated  twice.
\end{lemma}
Formally, the expansion  can be obtained by differentiation. The rigorous  proof is given in  Appendix A.
\begin{remark}
Note that asymptotics \eqref{int-rho-as} hold for any $\digamma\in C_0^{\infty}$ 
and the number of terms in the asymptotics depends only on the regularity of $p$ at $k=0$. 
For example, for $\rho$ satisfying \eqref{zero2} we obtain
\be\la{int-rho-as2}
\int \frac{|\hat\rho|^2 dk}{\hat D(i\mu+0)}=\langle |k|^4 \hat \rho, \frac{\hat\rho/|k|^4}{\hat D(i\mu+0)}\rangle
=\int \frac{|\hat\rho|^2 dk}{\hat D_0}+\sum\limits_{\ell=1}^{6} \mu^{\ell}\int \frac{P_{\ell}(k)|\hat\rho|^2 dk}{\hat D_0^{\ell}}
+{\cal O}(|\mu|^{7}),~~\mu\to 0.
\ee
\end{remark}
\begin{lemma}
Let condition \eqref{rosym}  hold.  Then ${\cal M}(i\mu)$ is invertible for sufficiently small $\mu\ne 0$, and 
\be\la{M-as}
{\cal M}^{-1}(i\mu)
=\begin{pmatrix} \ell_{11}\mu^{-2}+m_{11}+{\cal O}(|\mu|^2) & \ell_{12}\mu^{-1}+m_{12}\mu+ {\cal O}(|\mu|^3)\\
-\ell_{12}\mu^{-1}-m_{12}\mu+ {\cal O}(|\mu|^3)& \ell_{22}\mu^{-2}+m_{22}+{\cal O}(|\mu|^2)
\end{pmatrix},~~\mu\to 0,
\ee
where $\ell_{jj}<0$, $j=1,2$. The asymptotics  can be differentiated twice. 
\end{lemma}
\begin{proof}
By \eqref{PQFSR},  $V_{12}^{(k)}(0)=V_{21}^{(k)}(0)=0$ for  $k=0,2,4$, and 
$P_{jj}^{(k)}(0)=Q_{jj}^{(k)}(0)=F_{jj}^{(k)}(0)=U_{jj}^{(k)}(0)=0$ for  $k=1,3,5$ and  $j=1,2$.
Hence, condition \eqref{zero2}, formulas  \eqref{PQFSR}, \eqref{bb} and asymptotics \eqref{int-rho-as} imply
\be\la{a-j}
a_{jj}(i\mu+0)=-\alpha_{j}\mu^2+\beta_{j}\mu^4+{\cal O}(|\mu|^6),\quad 
a_{12}(i\mu+0)=\alpha_{12}\mu^3+\beta_{12}\mu^5+{\cal O}(|\mu|^7),\quad \mu\to 0,
\ee
where 
\beqn\nonumber
\alpha_{1}&=&\int\Big[\frac{v^2k_2^2k_1^2(k^2+3v^2k_1^2)}{\hat D_0^2}+k_2^2+\frac{4v^2k_1^2k_2^2}{\hat D_0}\Big]\frac{|\hat\rho|^2 dk}{k^2\hat D_0}
+{\omega}^2\int\frac{k_1^2|\na\hat\rho|^2(k^2+3v^2k_1^2)}{\hat D_0^3}dk+m>0,\\
\nonumber
\alpha_{2}&=&\int\Big[\frac{v^2k_2^4(k^2+3v^2k_1^2)}{\hat D_0^2}+k_1^2-\frac{4v^2k_1^2k_2^2}{\hat D_0}\Big]\frac{|\hat\rho|^2 dk}{k^2\hat D_0}
+{\omega}^2\int\frac{k_2^2|\na\hat\rho|^2(k^2+3v^2k_1^2)}{\hat D_0^3}dk+m>0,
\eeqn
since 
\[
4\frac{v^2k_1^2k_2^2}{\hat D_0}=2k_1\frac{2k_2^2k_1v^2}{D_0}\le  k_1^2+\frac{4k_2^4k_1^2v^4}{\hat D_0^2}
\le  k_1^2+\frac{k_2^4k_1^2v^2(1+3v^2)}{\hat D_0^2}\le  k_1^2+\frac{v^2k_2^4(k^2+3v^2k_1^2)}{\hat D_0^2}.
\]
Further,  formulas \eqref{fh-def}, \eqref{bb}, \eqref{PQFSR} and asymptotics \eqref{int-rho-as}  imply
\be\la{bb-ex}
b_{11}=\gamma_{1}\mu^2+\nu_{1}\mu^4+{\cal O}(|\mu|^6), 
\quad  b_{22}=\nu_{2}\mu^4+{\cal O}(|\mu|^6),\quad b_{12}=\gamma_{12}\mu^3+\nu_{12}\mu^5+{\cal O}(|\mu|^7),\quad \mu\to 0,
\ee
where
\beqn\nonumber
\gamma_{1}&=&\frac{{\omega}^2}{\kappa(0)}(h_1'(0))^2=
4\frac{{\omega}^2v^2}{I+\varkappa(0)}\Big(\int\frac{k_1^2|\na\hat\rho|^2}{\hat D_0^2}dk\Big)^2\\
\nonumber
&\le& \frac{4{\omega}^2 \varkappa(0)}{I+\varkappa(0)}\int\frac{k_1^4|\na\hat\rho|^2}{\hat D_0^3}dk 
\le {\omega}^2\int\frac{k_1^2|\na\hat\rho|^2(k^2+3v^2k_1^2)}{\hat D_0^3}dk<\alpha_{1}
\eeqn
by the Cauchy--Schwarz inequality and definition \eqref{fh-def} of $h$ and $\varkappa$.
Now \eqref{detM}  together with \eqref{a-j}--\eqref{bb-ex} imply that
$$
\det{\cal M}(i\mu)=(\al_{1}-\gamma_{1})\al_{2}\mu^4+((\beta_2+\nu_2)(\gamma_1-\al_1)-(\beta_1+\nu_1)\al_2+\al_{12}^2)\mu^6+{\cal O}(|\mu|^8),\quad \mu\to 0,
$$
where $(\alpha_{1}-\gamma_{1})\al_{2}>0$. Hence, ${\cal M}(i\mu)$ is invertible for sufficiently small $\mu\ne 0$.
Finally, \eqref{M-as} with  $\ell_{11}=-(\alpha_{1}-\gamma_{1})^{-1}$  and  $\ell_{22}=-\alpha_{2}^{-1}$   follows from  \eqref{detM}. 
\end{proof}
In order to have invertibility of ${\cal M}(i\mu)$ for all $|\mu|> 0$, we assume  the following  spectral condition:  
\be\la{M-condition} 
{\rm det}\, {\cal M}(i\mu)=[a_{11}(i\mu)+b_{11}(i\mu)][a_{22}(i\mu)+b_{22}(i\mu)]+[a_{12}(i\mu)+b_{12}(i\mu)]^2\ne 0~~{\rm for}~~\mu\ne 0.
\ee
In Appendix B we show that this condition holds i) for $v\ne 0$,  small $\omega$, and  large $I$, 
ii) for $v=0$ and small $\omega$ under the  Wiener  condition. 
\subsection{Asymptotics of ${\cal K}_0$ on the  imaginary axis}\label{pr-as}
Here we obtain asymptotics of the RHS ${\cal K}_0(i\mu+0)$ of the equation \eqref{M-def} as $\mu\to 0$ and $\mu\to\infty$.
\begin{lemma}\la{F-as}
Let $\rho$ satisfy  \eqref{rosym}, and let $(\Lam_0,\Psi_0)\in {\cal F}_{\beta}$ with $\beta>5/2$. Then
\be\la{F12-as}
{\cal K}_0(i\mu+0)=k_{2}\mu^2+{\cal O}(|\mu|^{5/2}),\quad \mu\to 0.
\ee
The asymptotics can be differentiated twice.
\end{lemma}
\begin{proof}
We will write ${\cal K}_0(i\mu)$, $\hat D(i\mu)$, etc.  instead of ${\cal K}_0(i\mu+i0)$, $\hat D(i\mu+i0)$.  
Formulas  \eqref{KK}, \eqref{F1}, and  \eqref{F1+} imply that
\beqn\nonumber
{\cal K}_0(i\mu)&=&\pi_0+[{\omega} PJ+im\mu]r_0+K_{01}(i\mu)+\mu K_{02}(i\mu)-\mu^2 \langle \hat\Lam_{0},\frac{\hat\rho}{\hat D(i\mu)}\rangle+K_{03}(i\mu)\\
\la{F-split}
&-&\frac{1}{\kappa(i\mu)}(K_{04}(i\mu)+K_{05}(i\mu))\big(i{\omega} h(i\mu)-\mu f(i\mu)-\breve P(i\mu)Jv\big),
\eeqn
where 
\beqn\nonumber
K_{01}(i\mu)&=&|v|^2\langle k\hat\Lam_{01},\frac{ k_1\hat\rho}{\hat D(i\mu)}\rangle-i|v|\langle \hat\Psi_{01},\frac{k\hat\rho}{\hat D(i\mu)}\rangle
+i\om|v|\langle k\cdot\! J\hat\Lam_{0},\frac{k_1\na\hat\rho}{\hat D(i\mu)}\rangle+\om\langle k\cdot\! J\hat\Psi_{0},\frac{\na\hat\rho}{\hat D(i\mu)}\rangle,\\
\nonumber
K_{02}(i\mu)&=&|v| \langle k\hat\Lam_{01},\frac{\hat\rho}{\hat D(i\mu)}\rangle+i{\omega}\langle k\cdot\! J\hat\Lam_{0},\frac{\na\hat\rho}{\hat D(i\mu)}\rangle
-|v| \langle k_1\hat\Lam_{0},\frac{\hat\rho}{\hat D(i\mu)}\rangle-i \langle \hat\Psi_{0}, \frac{\hat\rho}{\hat D(i\mu)}\rangle,\\
\nonumber
K_{03}(i\mu)&=&i|v|\langle\frac{(\widehat{{\cal P}r_0\rho})_1}{\hat D(i\mu)}, k\hat\rho\rangle+\om\langle\frac{Jk\cdot\! \widehat{{\cal P}r_0\rho}}{\hat D(i\mu)},\na\hat\rho\rangle
+i\mu\langle\frac{\widehat{{\cal P}r_0\rho}}{\hat D(i\mu)},\hat\rho\rangle,\\
\nonumber
K_{04}(i\mu)&=&i|v|\langle \frac{k_1\hat\Lam_0}{\hat D(i\mu)},J\hat\varrho\rangle+\mu\langle \frac{\hat\Lam_0}{\hat D(i\mu)},J\hat\varrho\rangle
+\langle \frac{\hat\Psi_0}{\hat D(i\mu)},J\hat\varrho\rangle,\quad 
K_{05}(i\mu)=-\langle \frac{\widehat{{\cal P}r_0\rho}}{\hat D(i\mu)},J\hat\varrho \rangle.
\eeqn
Applying  \eqref{int-rho-as}  with $\hat p=k_j\hat \rho$, $\hat p=k_j\na_{\ell}\hat \rho$, $\hat p=\hat\rho$, $\hat p=\na_j\hat\rho$  and $\hat p=\frac{k_j}{|k|^2}\hat\rho$  
and taking into account \eqref{zero2}, we  obtain
\beqn\la{High-as}
\!\!\!\!\!&&K_{01}(i\mu)=k_{10}+k_{11}\mu + k_{12}\mu^2 +{\cal O}(|\mu|^{5/2}),\quad K_{02}(i\mu)=k_{20}+k_{21}\mu +{\cal O}(|\mu|^{3/2}),\\
\la{High-as1}
\!\!\!\!\!&&\langle \hat\Lam_{0}, \frac{\hat\rho}{\hat D(i\mu)}\rangle=\sum\limits_{l=1,2}
\langle k_l\hat\Lam_{0},\frac{k_l\hat\rho}{|k|^2\hat D(i\mu)}\rangle=k_{02}+{\cal O}(|\mu|^{1/2}),\quad \mu\to 0.
\eeqn
Similarly to \eqref{High-as}--\eqref{High-as1}, we obtain that 
$K_{04}(i\mu)=k_{40}+k_{41}\mu +{\cal O}(|\mu|^{3/2})$.
Further, \eqref{Pi-e} and \eqref{int-rho-as2} imply
\be\la{High-as2}
K_{03}(i\mu)=k_{30}+k_{31}\mu +k_{32}\mu ^2+k_{33}\mu^3+{\cal O}(|\mu|^{4}).
\ee
Asymptotics  of type \eqref{High-as2} also hold for $K_{05}(i\mu)$, $f(i\mu)$, $h(i\mu)$, $P(i\mu)$, $\kappa(i\mu)$.
Hence, 
 \be\la{F12-as-full}
{\cal K}_0(i\mu+0)=k_0+k_1\mu+k_{2}\mu^2+{\cal O}(|\mu|^{5/2}),\quad \mu\to 0.
\ee
In Appendix C, we prove that the symplectic orthogonality conditions $\Omega(X_0,\tau_j)=0$ imply 
\be\la{F00}
k_0=k_1=0.
\ee
Combining \eqref{F12-as-full} with \eqref{F00}, we obtain \eqref{F12-as}. 
\end{proof}
\begin{lemma}\la{K0-as} 
Let $\rho$ satisfy  \eqref{rosym}, and let $(\Lam_0,\Psi_0)\in {\cal F}_{\beta}$ with $\beta>5/2$.  Then
\be\la{K0-rep}
{\cal K}_0(i\mu+0)=\pi_0+[{\omega} PJ+im\mu]r_0+\mathfrak{K}(i\mu+0),
\ee
where
\be\la{K-as}
|\mathfrak{K}(i\mu+0)|\le C(1+\Vert(\Lam_0,\Psi_0)\Vert_{{\cal F}_\beta}),~~
|\mathfrak{K}^{(k)}(i\mu+0)|\le \frac C{|\mu|}(1+\Vert(\Lam_0,\Psi_0)\Vert_{{\cal F}_\beta}),
\quad \mu\to\infty, \quad k=1,2.
\ee
\end{lemma}
\begin{proof}
Similarly to \eqref{q-as}, we obtain
\[
|\langle\frac{(\widehat{{\cal P}r_0\rho})_1}{\hat D(i\mu)}, k\hat\rho\rangle|\le C(\rho)|\mu|^{-2},\quad 
|\pa^j_{\mu}\langle\frac{(\widehat{{\cal P}r_0\rho})_1}{\hat D(i\mu)}, k\hat\rho\rangle|\le C(\rho)|\mu|^{-3},\quad
\mu\to\infty,\quad j=1,2.
\]
The same estimates also hold for $\langle\frac{Jk\cdot\! \widehat{{\cal P}r_0\rho}}{\hat D(i\mu)},\na\hat\rho\rangle$,
$\langle\frac{\widehat{{\cal P}r_0\rho}}{\hat D(i\mu)},\hat\rho\rangle$, and  
$\langle \frac{\widehat{{\cal P}r_0\rho}}{\hat D(i\mu)},J\hat\varrho\rangle$.
Further,  \eqref{R-as}  and \eqref{R-as1}  imply 
\beqn\nonumber
\!\!\!\!\!\!\!\!\!\!\!\!\!\!\!\!\!\!&&|\langle {\cal R}(i\mu+\!0)\Psi_0,p\rangle|\le C\Vert {\cal R}(i\mu+\!0)\Psi_0\Vert_{H^{-1}_{-\beta}}\Vert p\Vert_{H^1_{\beta}}
\le C |\mu|^{-2}\Vert\Psi_0\Vert_{L^{2}_{\beta}},\\
\la{Phi-as}
\!\!\!\!\!\!\!\!\!\!\!\!\!\!\!\!\!\!&&|\langle {\cal R}^{(k)}(i\mu+\!0)\Psi_0,p\rangle|\le C\Vert {\cal R}^{(k)}(i\mu+\!0)\Psi_0\Vert_{H^{-2}_{-\beta}}\Vert p\Vert_{H^2_{\beta}}
\le C |\mu|^{-3}\Vert\Psi_0\Vert_{L^{2}_{\beta}},~~\mu\to\infty,~k\!=\!1,2,
\eeqn
where $p$ is one of the functions $\rho$, $\na_j\rho$,  $y_l\rho$, $y_l\na_j\rho$, $j,l=1,2$.
The same estimates hold with  $\na_k\Lam_0$ instead of $\Psi_0$. Moreover, \eqref{R-as}  and \eqref{R-as1}  imply  for $k=1,2$
 \beqn\nonumber
\!\!\!\!\!\!\!\!\!\!\!\!\!\!\!\!\!\!&&|\langle {\cal R}(i\mu+\!0)\Lam_0,p\rangle|
\le C\sum\limits_{j=1,2}\Vert {\cal R}(i\mu+\!0)\na_j\Lam_0\Vert_{H^{-1}_{-\beta}}\Vert \eta_j\Vert_{H^1_{\beta}}
\le C |\mu|^{-2}\Vert\na\Lam_0\Vert_{L^{2}_{\beta}},\\
\la{Phi-as1}
\!\!\!\!\!\!\!\!\!\!\!\!\!\!\!\!\!\!&&|\langle {\cal R}^{(k)}(i\mu+\!0)\Lam_0,p\rangle|
\le C\sum\limits_{j=1,2}\Vert {\cal R}^{(k)}(i\mu+\!0)\na_j\Lam_0\Vert_{H^{-2}_{-\beta}}\Vert \eta_j\Vert_{H^2_{\beta}}
\le C |\mu|^{-3}\Vert\na\Lam_0\Vert_{L^{2}_{\beta}},~~\mu\to\infty.
\eeqn
Here $\eta_j=F^{-1}_{k\to x}\frac {k_j\hat p(k)}{|k|^2}$.
Hence, \eqref{K-as}  follows from  \eqref{F-split}. 
\end{proof}
\subsection{Proof of Theorem \ref{lindecay}}\label{X-decay}
{\bf Time decay of the vector components}\\
Applying the inverse Fourier--Laplace transform, we obtain
\beqn\nonumber
r(t)&\!\!=\!\!&\frac{1}{2\pi}\int e^{i\mu t} \ti r(i\mu+\!0)d\mu
=\frac{1}{2\pi}\int e^{i\mu t}\zeta(\mu) \ti r(i\mu+\!0)d\mu+\frac{1}{2\pi}\int e^{i\mu t}(1\!-\!\zeta(\mu)) \ti r(i\mu+\!0)d\mu\\
\la{F-int}
&\!\!=\!\!&I_1(t)+I_2(t).
\eeqn
Here $\zeta\in C_0^{\infty}(\R)$,  ${\rm supp}\,\zeta\subset (-2,2)$, $\zeta(\mu)=1$ for $|\mu| \le 1$.
Using  \eqref{M-as-} and \eqref{K0-rep}--\eqref{K-as}, we obtain 
\be\la{tr-decay}
\ti r(i\mu+\!0)={\cal M}^{-1}(i\mu+0){\cal K}_0(i\mu+0)=\frac{r_0}{i\mu}+{\cal O}(|\mu|^{-2}),\quad 
\ti r^{(k)}(i\mu+\!0)={\cal O}(|\mu|^{-1-k}), ~~ \mu\to\infty
\ee
for $k=1,2$. Moreover, $\ti r(i\mu+\!0)\in C^2(\R\setminus 0)$ by \eqref{M-condition}.
Therefore, double integration  by parts  gives
\be\la{I2}
I_2(t)={\cal O}(t^{-2}),\quad t\to \infty.
\ee
Now consider $I_1(t)$. Asymptotics \eqref{M-as} and \eqref{F12-as} imply
\be\la{zyg}
\ti r^{(k)}(i\mu+\!0)={\cal O}(|\mu|^{1/2-k}),\quad\mu\to 0,\qquad k=0,1,2. 
\ee
Applying  the known result on ``one-and-a-half  integration by parts'' (see \cite[Lemma 10.2]{JK79}, \cite[Lemma 22.5]{KK2012}), we  obtain
\be\la{I1}
I_1(t)={\cal O}(t^{-3/2}),\quad t\to \infty.
\ee
Finally, \eqref{I2} and \eqref{I1} imply
\be\la{r-decay}
r(t)={\cal O}(t^{-3/2}),\quad t\to \infty.
\ee
Further, the third equation of \eqref{eq-main} together with \eqref{tr-decay} gives
\[
\ti\phi(i\mu+0)=\frac{\omega}m JP\ti r+i\mu\ti r-r_0=\frac{\omega}{im\mu} JPr_0+{\cal O}(|\mu|^{-2}),\quad \mu\to \infty.
\]
Moreover,  $\ti\phi^{(k)}(i\mu+0)={\cal O}(|\mu|^{1/2-k})$ as $\mu\to 0$ for $k=0,1,2$, by \eqref{eq-main} and \eqref{zyg}. Hence, similarly to \eqref{r-decay},
\be\la{Phi-decay}
\phi(t)=\frac{1}{2\pi}\int e^{i\mu t} \ti\phi(i\mu+0) d\mu={\cal O}(t^{-3/2}), \quad t\to \infty.
\ee
Below we also need a decay for $\nu(t)$. Definition \eqref{nu-rep} of $\ti\nu$, asymptotics of type \eqref{q-as} for $\varkappa$, $f$, $h$,  and $P$, and asymptotics  \eqref{Phi-as}--\eqref{Phi-as1} imply 
 $$
 \ti\nu(i\mu+0)={\cal O}(|\mu|^{-1}),\quad \ti\nu^{(k)}(i\mu+0)={\cal O}(|\mu|^{-2}),\qquad \mu\to\infty,\quad k=1,2.
 $$ 
Moreover,  $\ti \nu^{(k)}(i\mu+\!0)={\cal O}(|\mu|^{1/2-k})$ as $\mu\to 0$ for $k=0,1,2$ 
by \eqref{nu-rep}, \eqref{int-rho-as2}, \eqref{High-as}, and \eqref{zyg}. Hence,
\be\la{nu-decay}
\nu(t)=\frac{1}{2\pi}\int e^{i\mu t} \ti\nu(i\mu+0) d\mu={\cal O}(t^{-3/2}), \quad t\to \infty.
\ee
{\bf Time decay of the fields}\\
Denote ${\cal F}(t)=(\Lam(t),\Psi(t))$ and rewrite the first two equations of \eqref{bfA} as
\be\la{F-eq}
\dot{\cal F}(t)=\begin{bmatrix} v\cdot\!\na &1\\ 
\Delta& v\cdot\!\na\end{bmatrix}{\cal F}(t)+\begin{bmatrix} 0\\ Q(t)\end{bmatrix}, \quad
Q={\cal P}[\rho\phi - \frac {\omega}m\rho J Pr-\!vr\cdot\!\na\rho]+\big({\omega}r\cdot\! \na-\!\nu\big)J\varrho.
\ee
Applying the Duhamel representation, we get
\be\la{Duhamel}
{\cal F}(t)=W_v(t){\cal F}_0+\int_0^t W_v(t-s)\begin{bmatrix} 0\\ Q(s)\end{bmatrix} ds,\quad t\ge 0,
\ee
where $W_v(t)$ is the dynamical group of the modified wave equation (equation \eqref{F-eq} with $Q(t)=0$).
For the group  $W_v(t)$ the dispersive decay holds \cite{K2010, KK2023}:
\be\la{dede}
\Vert W_v(t){\cal F}_0\Vert_{{\cal F}_{-\beta}}\le C(1+t)^{-2}\Vert{\cal F}_0\Vert_{{\cal F}_\beta},\quad \beta>2,\quad |v|<1,\quad t\ge 0.
\ee
 Applying  \eqref{dede} to \eqref{Duhamel} and  using   \eqref{r-decay}--\eqref{nu-decay}, we obtain \eqref{frozenest} 
 for the field components.\\
It remains to prove the decay \eqref{frozenest} for $\pi(t)$ and $\psi(t)$.  The decay for  $\pi(t)$ follows from 
\eqref{nu-def}. Finally, symplectic  orthogonality condition
$\Om(Z(t),\tau_6(v,\omega))=0$ together with \eqref{Mvw} gives
$$
\psi(t)=\frac{1}{\kappa(0)}\big[\langle\Lam(t),\pa_{\omega}\Pi_{v,{\omega}}\rangle-\langle\Psi(t),\pa_{\omega}A_{v,\omega}\rangle\big]
={\cal O}(t^{-3/2}),\quad t\to \infty.
$$
\setcounter{equation}{0}
\section{Appendix}
\subsection*{A. Proof of Lemma \ref{Dr-as}}
By \eqref{g-rep1},
\be\la{g-pm}
{\cal R}(i\mu+0,z)=\gamma e^{-i|v|\nu \tilde z_1}R_\mp(\nu^2,\tilde z),\quad \pm\mu>0,\quad z=x-y, \quad \nu=\gamma\mu.
\ee
Recall that (cf. \cite{K2010, S2005}) 
$$
R_{\pm}(\zeta^2,z)=\pm\frac{i}{4} H_0^{\pm}(\zeta|z|),\quad\zeta>0,
$$
where $H_0(y)=J_0(y)+iY_0(y)$ is the Hankel function.
For $|\nu z|\le 1$, the  well-known asymptotics of $H_0^{\pm}$ (cf. \cite[Formula 10.8]{O}) imply   that
\be\la{g-as-s}
R_{\pm}(\nu^2,\tilde z)=-\frac{\log (|\nu \tilde z|)}{2\pi} + h_0^{\pm}+\frac{(\nu|\tilde z|)^2}{8\pi}\log(\nu |\ti z|)
+h_1^{\pm}(\nu|\tilde z|)^2-\frac{(\nu|\tilde z|)^{4}\log (|\nu \tilde z|)}{128\pi}+\dots,
\ee
where $h_j^{\pm}=h_j^{\pm}(\gamma)$  are  some  constants. 
Moreover, 
\be\la{g-as-s+}
e^{-i|v|\nu \tilde z_1}=1-i |v|\nu\tilde z_1-\frac{v^2}2\nu^2\tilde z_1^2+i\frac{|v|^3}6\nu^3\tilde z_1^3+\frac{ v^4}{24}\nu^4\tilde z_1^4+\dots
\ee
Hence,  \eqref{g-pm}--\eqref{g-as-s+}  imply, for $|\nu \tilde z|\le 1$,
\beqn\nonumber
\!\!\!&&\frac{1}{\gamma}{\cal R}(i\mu+0,z)=-\frac{1}{2\pi}\log (|\nu\tilde z|) +\tilde h_0^{\pm}+\frac{i |v|}{2\pi}\nu\tilde z_1\log (|\nu\tilde z|)
+\tilde h_1^{\pm}\nu\tilde z_1+\frac{1}{8\pi}\nu^2\log (|\nu\tilde z|)\big(2v^2\tilde z_1^2+|\tilde z|^2\big)\\
\nonumber
\!\!\!&&+\nu^2(\tilde h_{21}^{\pm}\tilde z_1^2+\tilde h_{22}^{\pm}|\tilde z|^2)
-\frac{i |v|}{24\pi}\nu^3\log(|\nu\tilde z|)\big(3\tilde z_1|\tilde z|^2+2v^2\tilde z_1^3\big)
+\nu^3(\tilde h_{31}^{\pm}\tilde z_1|\tilde z|^2+\tilde h_{32}^{\pm}\tilde z_1^3)+\dots\\
\la{g+}
\!\!\!&&= P_j^{\pm}(\mu, z)+Q_j^{\pm}(\mu,z),\quad \pm\mu>0,\quad j\ge 2.
\eeqn
Here  $P_j^{\pm}(\mu, z)$ is the sum of terms with powers of $\nu$ less than $j$ and $Q_j^{\pm}(\mu, z)$  is the remainder.
In particular,
$P_2^{\pm}(\mu, z)=-\frac{1}{2\pi}\log (|\nu\tilde z|) +\tilde h_0^{\pm}+\frac{i|v|}{2\pi}\nu\tilde z_1\log (|\nu\tilde z|)+\tilde h_1^{\pm}\nu\tilde z_1$.
The remainder $Q_j^{\pm}(\mu,z)$  is estimated as
\be\la{Q-est}
|\pa_{\mu}^{k}Q_j^{\pm}(\mu, z)|\le C_j|\ti z|^k(|\nu\ti z|)^{j-k}|\log(|\nu\tilde z|)|
\le C_j|z|^{j-\frac 12}|\mu|^{j-k-\frac 12},\quad  |\nu\ti z|\le 1,\quad k=0,1,2.
\ee
Further, for $|\nu \tilde z|\ge 1$, one has
$$
|\pa_{\nu}^k R_{\pm}(\nu^2,\tilde z)|=\frac 14|\pa_{\nu}^k H_0^{\pm}(|\nu\tilde z|)|\le \frac{C|\tilde z|^k}{\sqrt{|\nu \tilde z|}},\qquad
|\pa_{\nu}^k e^{-i|v|\nu \tilde z_1}|\le C|\tilde z_1|^k,\quad k=0,1,2.
$$
Hence, for $|\nu \tilde z|\ge 1$ and $j\ge 2$, 
\be\la{g-as-s1}
|\pa_{\mu}^k {\cal R}(i\mu\!+0,z)|\le \frac{C|\tilde z|^k}{\sqrt{|\nu \tilde z|}}\le C_j |\ti z|^k(|\nu \ti z|)^{j-k-\frac 12}=C_j |z|^{j-\frac 12}|\mu|^{j-k-\frac 12}, 
~~ k=0,1,2.
\ee
Moreover, for  $|\nu \tilde z|\ge 1$ and  $j\ge 2$,
\be\la{P-est}
|\pa_{\mu}^kP_j^{\pm}(\mu,z)|\le C |\pa_{\nu}^k \big((\nu\ti z)^{j-1}\log (|\nu\tilde z|)\big)|\ 
\le C_{j} |z|^{j-\frac 12}|\mu|^{j-k-\frac 12},\quad k=0,1,2.
\ee
Finally,  \eqref{g+}--\eqref{P-est}  imply  that 
\be\la{fin}
|\pa_{\mu}^k Q_j^{\pm}(\mu,z)|\le C_j |z|^{j-\frac 12}|\mu|^{j-k-\frac 12}, \quad k=0,1,2,\quad j\ge 2.
\ee
 We prove   Lemma \ref{Dr-as} in the case $j=2$ only.  The other cases $j\ge 3$ are considered similarly.
We assume that $p\in C_0^\infty(\R^2)$ satisfies  $\hat p(k)={\cal O}(|k|^4)$ as $k\to 0$. 
Then \eqref{g+} implies that
\be\la{gp}
 {\cal R}(i\mu+0)p=g_0p+\mu g_1 p +\mu^2 g_2 p+\gamma Q_3^{\pm}(\mu)p.
\ee
Here $g_j$ are integral operators with  kernels 
$$
g_0(x,y)=-\frac{\gamma}{2\pi}\log|\tilde z|, \quad
g_1(z)=\frac{i\gamma^2 |v|(\tilde z_1)\log|\tilde z|}{2\pi},\quad 
g_2(z)=-\frac{\gamma^3 \log|\tilde z|}{8\pi}(2v^2\ti z_1^2+|\ti z|^2),\quad  \ti z=\ti x-\ti y.
$$
In \cite{KK2023} it was proved that 
$$
\widehat {g_0 p}(k)=\frac{\hat p(k)}{\hat D_0(k)},\quad \widehat {g_1p}(k)=\frac{2|v|k_1\hat p(k)}{\hat D_0^2(k)},
\quad  \widehat {g_2p}(k)=\frac{(k^2+3v^2k_1^2)\hat p(k)}{\hat D_0^3(k)},
$$
and these Fourier transforms are bounded at $k=0$.
Denote by $q$ the function with  $\hat q(k)=\frac{i\hat p(k)}{k_1}$, so that   $p(x)=\na_1 q(x)$.
One has
$$
[Q_3^{\pm}(\mu)p](x) 
=\int Q_3^{\pm}(\mu, x-y)\na_1 q(y)dy=-\int \na_1 Q_3^{\pm}(\mu,x-y) q(y)dy.
$$
Note that   differentiation reduces the powers of $z=x-y$ in $Q_3^{\pm}(\mu, z)$.
Namely,
$$
|\pa_{\mu}^k(\na_1 Q_3(\mu, z))|\le C|z|^{3/2}|\mu|^{5/2-k}.
$$ 
Hence,  \eqref{int-rho-as} follows for  any  $\digamma\in L^2_{\beta}$ with $\beta>5/2$.
\subsection*{B. Feasibility  of the spectral condition}
Here we provide examples illustrating that condition \eqref{M-condition}  can be satisfied.
First we prove an auxiliary lemma.
\begin{lemma}\la{ek1}
 Let $P(k)$ be a real polynomial, $P(k)\ge 0$, $k\!\in\R^2$ and $P(k)\not\equiv 0$.   Then for  $0\!<|v|<\!1$,
\be\la{W}
\rIm \Big[\int\frac{P(k)|\hat\rho(k)|^2\,dk}{k^2\hat D(i\mu+0)}\Big]\lessgtr 0~~{\rm if} ~~ \mu \gtrless 0.
\ee
\end{lemma}
\begin{proof}
We assume that  $v=(|v|,0)$. In this case,   $\hat D(i\mu+0)=k^2+(i|v|k_1+(i\mu+0))^2$. Denote
$K_v(\mu)=\{k\in\R^2: (k_1-\mu |v|\gamma^2)^2+k_2^2\gamma^2=\mu^2\gamma^4\}$, $\gamma=1/\sqrt{1-v^2}$. By the Plemelj formula,
\be\la{PF}
\rIm \Big(\int\frac{P(k)|\hat\rho(k)|^2\,dk}{k^2\hat D(i\mu+0)}\Big)=-\frac{\mu}{|\mu|}\pi\int_{K_v(\mu)} \frac{P(k)|\hat\rho(k)|^2}{k^2|\na_k\hat D(i\mu+0)|}dS
\ee
(see, for example, \cite[Lemma 15.3]{IKV2006}).
Hence, it suffices to show that  $\hat\rho(k) \not{\!\!\equiv}\, 0$ on $K_v(\mu)$ for $\mu\ne 0$.
Assume, on the contrary, $\hat\rho(k) \equiv 0$ on $K_v(\mu)$.
Then  $\hat\rho(k) \equiv 0$ in the ring  $ |\mu|(1-|v|)\gamma^2 \le |k|\le |\mu|(1+|v|)\gamma^2$  by  (\ref{rosym}). Hence,
$\hat\rho(k) \equiv 0$ in $\C^2$, due to the analyticity of $\hat\rho(k)$,
in contradiction  to the last condition in \eqref{rosym}.
\end{proof}
Let $0\!<|v|<\!1$. By Lemma  \ref{ek1},  $\rIm\varkappa(i\mu+0)=\rIm \int\frac{|\na\hat\rho|^2}{\hat D(i\mu+0)}dk\ne 0$ for $\mu\ne 0$.
Let us  also  assume  that $I=\infty$ and  $\omega=0$. Then 
${\rm \det}\,{\cal M}(i\mu+0)=a_{11}(i\mu+0)a_{22}(i\mu+0)$, where
\[
\rIm a_{11}=\rIm \!\int\!\frac{(|v|k_1k_2+\mu k_2)^2|\hat\rho|^2}{k^2\hat D(i\mu+0)}dk\ne 0,
\quad \rIm a_{22}=\rIm \!\int\!\frac{(|v|k_2^2-\mu k_1)^2|\hat\rho|^2}{k^2\hat D(i\mu+0)}dk\ne 0,\quad \mu\ne 0.
\]
by \eqref{PF}. Thus, condition \eqref{M-condition}  holds  for  $0\!<|v|<\!1$,   $I=\infty$,  and  $\omega=0$.  
Hence, it holds for  large $I$ and   small $|\omega|$ by the continuity of $\det{\cal M}(i\mu+0)$.
\smallskip\\
In the case  $v=0$,  formulas  \eqref{fh-def}, \eqref{PQFSR},  and \eqref{bb} imply that   $a_{22}(i\mu+0)=a_{11}(i\mu+0)$, and 
\be\la{W2}
{\rm det}\, {\cal M}(i\mu+0)=\big(a_{11}(i\mu+0)+ia_{12}(i\mu+0)\big)\big(a_{11}(i\mu+0)-ia_{12}(i\mu+0)\big),
\ee
where
\beqn\nonumber
a_{11}(i\mu)\pm ia_{12}(i\mu)&=&-\mu^2\big(m+\frac 12\int\! \frac{|\hat\rho(k)|^2dk}{k^2+(i\mu+0)^2}\big)
+\frac {\om^2}2\big(\int\! \frac{k^2|\na\hat\rho(k)|^2dk}{k^2}-\int\! \frac{k^2|\na\hat\rho(k)|^2dk}{k^2+(i\mu+0)^2}\big)\\
\la{W3}
&\pm&\om\mu\big(\int\! \frac{\hat\rho(k) (k\cdot\na)\hat\rho(k)dk}{k^2}-\int\! \frac{\hat\rho(k) (k\cdot\na)\hat\rho(k)dk}{k^2+(i\mu+0)^2}\big).
\eeqn
The Wiener condition and the Plemelj formula imply
$$
\rIm(a_{11}(i\mu)\pm ia_{12}(i\mu))=-\frac{\mu^2}2 \rIm \int\! \frac{|\hat\rho(k)|^2dk}{k^2+(i\mu+0)^2}\ne 0~~{\rm for}~~\om=0~~{\rm and}~~ \mu\ne 0.
$$
Hence, in the case $v=0$, ${\rm det}\, {\cal M}(i\mu)\ne 0$ for $\mu\ne 0$ and all sufficiently small $|\omega|$.
\subsection*{C. Symplectic orthogonality conditions}
Here we prove \eqref{F00}. First we prove that  $k_0={\cal K}_0(0)=0$.  By \eqref{inb} and \eqref{APi}, 
\beqn\nonumber
0&\!\!\!=\!\!\!&\Omega(X_0,\tau_j)=-\langle\Lam_0,\pa_j \Pi\rangle+\langle\Psi_0,\pa_j A\rangle-\pi_0\cdot e_j
=-v^2\!\int\! \hat\Lam_{01} \frac{k_jk_1\rho}{\hat D_0}dk+i|v|\!\int\! \hat\Psi_{01}\frac{k_j\rho}{\hat D_0} dk\\
\la{so-con}
&\!\!\!+\!\!\!&i{\omega}|v|\!\int\! (\hat\Lam_0\cdot \!J\na\hat\rho) \frac{k_jk_1}{\hat D_0} dk
+{\omega}\!\int\! (\hat\Psi_0\cdot\!  J\na\hat\rho)\frac{k_j}{\hat D_0} dk -\pi_{0j},\quad j=1,2.
\eeqn
On the other hand,  \eqref{F1} and \eqref{F-split}   imply
\beqn\nonumber
&&{\cal K}_0(0)=\pi_0+{\omega} PJr_0+v^2\int\frac{\hat\Lam_{01}k_1\hat\rho}{\hat D_0} k dk-i|v|\int\frac{\hat\Psi_{01}\hat\rho}{\hat D_0} k dk\\
\la{soc1}
&&+i|v|{\omega}\int\frac {k_1(k\cdot J\hat \Lam_0)}{\hat D_0} \na\hat\rho dk
+{\omega}\int\frac {k\cdot J\hat \Psi_0}{\hat D_0} \na\hat\rho dk+\om\langle\frac{Jk\cdot\! \widehat{{\cal P}r_0\rho}}{\hat D_0},\na\hat\rho\rangle.
\eeqn
Comparing \eqref{so-con} and \eqref{soc1}, we obtain that  $k_{0j}=-\Omega(X_0,\tau_j)=0$, $j=1,2$, since
$$
\langle\frac{Jk\cdot\! \widehat{{\cal P}r_0\rho}}{\hat D_0},\na\hat\rho\rangle=\langle\frac{(r_0\cdot Jk)\rho}{\hat D_0},\na\hat\rho\rangle=- PJr_0
$$
by \eqref{Pi-e} and \eqref{c-ro}.
It remains to prove that $k_{1}={\cal K}_0'(0)=0$. By \eqref{inb} and \eqref{Mvw},
\[
\Omega(X_0,\tau_{6})=\langle\Lam_0,\pa_{\om}\Pi\rangle-\langle\Psi_0,\pa_{\om}A\rangle+\psi_0\pa_{\om}M
=-\langle\Lam_0,\frac{i(vk)J\hat\varrho}{\hat D_0}\rangle
+\langle\Psi_0,\frac{J\hat\varrho}{\hat D_0}\rangle+\psi_0(I+\varkappa(0))=0.
\]
Hence,
\be\la{soc}
\psi_0=-\frac{1}{I+\varkappa(0)}\langle \frac{\Phi_0}{\hat D_0}, J\hat\varrho\rangle.
\ee
Further, \eqref{solY} and \eqref{inb} imply
\be\la{X0-2}
0=\Omega(X_0,\tau_{j+2})=\langle\Lam_0, \pa_{v_j} \Pi\rangle-\langle\Psi_0,\pa_{v_j} A\rangle+r_0\cdot (me_j+\langle\pa_{v_j} A,\rho\rangle)
+\psi_0\pa_{v_j}M,\quad j=1,2.
\ee
By  \eqref{solit3} and \eqref{solit32+}--\eqref{A1A1++}, 
\beqn\nonumber
\!\!\!\!\!&&\!\!\!\!\!\langle\Psi_0,\pa_{v_j} A\rangle=\int\Big(\frac{\hat\Psi_{0j}}{\hat D_0}+2v^2\frac{ \hat\Psi_{01}k_1k_j}{\hat D_0^2}\Big)\hat\rho dk
-2i|v|{\omega}\int \frac{k_1k_j(\hat\Psi_0\cdot J\na\hat\rho)}{\hat D_0^2}dk,\\
\nonumber
\!\!\!\!\!&&\!\!\!\!\!\langle\Lam_0, \pa_{v_j} \Pi\rangle=-i|v|\!\int\!\Big[\frac{\hat\Lam_{01}k_j}{\hat D_0}
+\frac{\hat\Lam_{0j}k_1}{\hat D_0}+2v^2\frac{\hat\Lam_{01}k_1^2k_j}{\hat D_0^2}\Big]\hat\rho dk
-{\omega}\!\int \!\frac{k_j}{\hat D_0}\Big[1+\frac{2v^2k_1^2}{\hat D_0}\Big](\hat\Lam_0\cdot J\na\hat\rho)dk,
\eeqn
\[
\langle \pa_{v_1}A,\rho\rangle=e_1\int\frac{(v^2k_1^2+k^2)k_2^2|\hat\rho|^2}{k^2\hat D_0^2}dk, \quad
\langle \pa_{v_2}A,\rho\rangle=-e_2\int\frac{(v^2k_2^2+k^2(v^2-1))k_1^2|\hat\rho|^2}{k^2\hat D_0^2}dk.
\]
Substituting these equations   into \eqref{X0-2} and taking into account \eqref{Mvw} and \eqref{soc}, we obtain
\beqn\nonumber
&&\!\!\!\!\!\!\!0=-i|v|\!\int\! \frac{\hat\Lam_{01}(k^2+v^2k_1^2)\hat\rho}{\hat D_0^2}k dk-i|v|\!\int\! \frac{k_1\hat\rho}{\hat D_0}\Lam_{0} dk
-{\omega}\!\int\! (\hat\Lam_0\cdot J\na\hat\rho) \frac{k^2+v^2k_1^2}{\hat D_0^2}kdk-\!\int\! \frac{\hat\rho}{\hat D_0}\Psi_{0} dk\\
\nonumber
&&\!\!\!\!\!\!\!-2v^2\!\int \frac{ \hat\Psi_{01}k_1\hat\rho}{\hat D_0^2}kdk+2i|v|{\omega}\!\int (\hat\Psi_0\cdot J\na\hat\rho) \frac{k_1}{\hat D_0^2}kdk
+m r_{0}+r_{01}e_1\int\frac{(v^2k_1^2+k^2)k_2^2|\hat\rho|^2}{k^2\hat D_0^2}dk\\
\la{soc2}
&&\!\!\!\!\!\!\!
-r_{02}e_2\int\frac{(v^2k_2^2+k^2(v^2-1))k_1^2|\hat\rho|^2}{k^2\hat D_0^2}dk
-2e_1\frac{|v|{\omega}}{I+\varkappa(0)}\langle \frac{\Phi_0}{\hat D_0}, J\hat\varrho\rangle \!\int\frac{k_1^2|\na\hat\rho|^2}{\hat D_0^2} dk. 
\eeqn
Denote 
$
T_0:=\pa_\lam\Big(\frac{\Phi_0(\lam)}{\hat D(\lam)}\Big)\Big|_{\lam=0}=\ds\frac{(k^2+v^2k_1^2)\Lam_{0}-2i|v|k_1\hat\Psi_0} {\hat D^2_0}+\frac{2i|v|k_1\hat\rho(r_0\cdot Jk)Jk} {k^2\hat D^2_0}.
$
Then  \eqref{F1},  \eqref{fh-def}, \eqref{F-split} and \eqref{soc2}  imply
\beqn\nonumber
k_1={\cal K}_0'(0)\!\!&\!\!=\!\!&\!\!-i|v|\langle T_0,k\hat\rho\rangle
+{\omega}\langle k\cdot J T_0 ,\na\hat\rho\rangle-\langle \Phi_0,\hat\rho\rangle
+mr_0-i{\omega}\frac{\langle \Phi_0(0),J\hat\varrho\rangle}{1+\varkappa(0)}h'(0)\\
\nonumber
\!\!&\!\!=\!\!&\!\!-i|v|\langle\frac{(k^2\!+\!v^2k_1^2)\Lam_{01}} {\hat D^2_0},k\hat\rho\rangle
-2v^2\langle\frac{k_1\Psi_{01}} {\hat D^2_0},k\hat\rho\rangle
+2v^2\langle \frac{k_1k_2\hat\rho(r_0\cdot Jk)} {k^2\hat D^2_0},k\hat\rho\rangle+mr_0\\
\nonumber
\!\!&\!\!+\!\!&\!\!{\omega}\Big[\langle \frac{(k^2+v^2k_1^2)(k\cdot J\Lam_0)}{\hat D^2_0},\na\hat\rho\rangle
-2i|v|\langle \frac{k_1(k\cdot J\Psi_0)}{\hat D^2_0},\na\hat\rho\rangle-2i|v|\langle \frac{k_1\hat\rho(r_0\cdot Jk)}{\hat D^2_0},\na\hat\rho\rangle\Big]\\
\nonumber
\!\!&\!\!-\!\!&\!\!i|v|\langle \frac{k_1\Lam_0}{\hat D_0},\hat\rho\rangle\! -\!\langle \frac{\Psi_0}{\hat D_0},\hat\rho\rangle
\!+\!\langle\frac{\hat\rho(r_0\cdot Jk)}{k^2\hat D_0}Jk,\hat\rho\rangle
-2e_1{\omega}|v|\frac{\langle \Phi_0(0),J\hat\varrho\rangle}{1+\varkappa(0)}\!\!\int\!\! \frac{k_1^2|\na_1\rho|^2}{\hat D_0^2}dk=0,
\eeqn
since
\beqn\nonumber
&&2v^2\langle\frac{ k_1k_2\hat\rho(r_0\cdot Jk)}{k^2\hat D_0^2},k\hat\rho\rangle+\langle\frac{\hat\rho(r_0\cdot Jk)}{k^2\hat D_0}Jk,\hat\rho\rangle\\
\nonumber
&&~~~~=r_{01}e_1\int\frac{(v^2k_1^2+k^2)k_2^2|\hat\rho|^2}{k^2\hat D_0^2}dk
-r_{02}e_2\int\frac{(v^2k_2^2+k^2(v^2-1))k_1^2|\hat\rho|^2}{k^2\hat D_0^2}dk.
\eeqn



\begin{thebibliography}{99}
\bibitem{A}
S. Agmon, Spectral properties of Schr\"odinger operator and scattering theory,
{\em Ann.~Scuola Norm.~Sup.~Pisa}, Ser.~IV {\bf 2}, 151-218 (1975).
\vspace{-3 mm}
\bibitem{BKKS}
V. Buslaev, A. Komech, E. Kopylova, D. Stuart,
On Asymptotic Stability of Solitary Waves in Schr\"odinger Equation Coupled to Nonlinear Oscillator,
{\em Comm. Partial Diff. Eqns.}  {\bf 33} (2008), no. 4, 669--705.
\vspace{-3 mm}
\bibitem{BS}
V.S. Buslaev, C. Sulem, On asymptotic stability of solitary waves for nonlinear Schr\"odinger equations,
{\em Ann. Inst. Henri Poincar\'e, Anal. Non Lin\'eaire} {\bf 20}(2003), no.3, 419-475.
\vspace{-3 mm}
\bibitem{CKK2023}
A. Comech, A. Komech, E. Kopylova,  Attractors of Hamiltonian nonlinear partial differential equations. 
Chapter in: Partial Differential Equations and Functional Analysis, Birkhauser, 2023.
\vspace{-3 mm}
\bibitem{IKM2004}
V. Imaykin, A. Komech, N. Mauser, 
Soliton-type asymptotics for the coupled Maxwell--Lorentz equations, {\em Ann. Inst. Poincar\'e  Phys. Theor.} {\bf  5} (2004), 1117--1135.
\vspace{-3 mm}
\bibitem{IKS2004}
V. Imaykin, A. Komech, H. Spohn, 
Rotating charge coupled to the Maxwell field: scattering theory and adiabatic limit, 
{\em Monatsh. Math.} {\bf 142} (2004), no. 1--2,  143--156.
\vspace{-3 mm} 
 \bibitem{IKS2011}
V. Imaykin, A. Komech, H. Spohn, 
Scattering asymptotics for a charged particle coupled to the Maxwell field, 
{\em J. Math. Phys.} {\bf 52} (2011), no. 4,  042701.
\vspace{-3 mm}
\bibitem{IKV2006}
V. Imaykin, A. Komech, B. Vainberg,
On scattering of solitons for the Klein--Gordon equation coupled to a particle
{\em Commun. Math. Phys.} {\bf 268} (2006), 321--367. 
\vspace{-3 mm}\bibitem{IKV2012}
V. Imaykin, A. Komech, B. Vainberg, 
 Scattering of solitons for coupled wave-particle equations,  {\em J. Math. Analysis and Appl.} {\bf 389} (2012), no. 2, 713--740.
\vspace{-3 mm}
\bibitem{JK79}
Jensen A., Kato T., 
Spectral properties of Schr\"odinger  operators and time-decay of the wave functions, 
{\em Duke Math. J.} {\bf 46} (1979), 583--611.
\vspace{-3 mm}
\bibitem{KS2000}
A. Komech, H. Spohn, 
Long-time asymptotics for the coupled Maxwell- Lorentz equations, 
{\em Comm. Partial Diff. Eqs.} {\bf 25} (2000), no. 3/4, 559--584.
\vspace{-3 mm}
\bibitem{KK2006}
A. Komech, E.A. Kopylova, 
Scattering of solitons for Schr\"odinger equation coupled to a particle, 
{\em Russian J. Math. Phys.} {\bf 50} (2006), no. 2, 158--187. 
\vspace{-3 mm}
\bibitem{KK2012}
A. Komech, E. Kopylova, 
Dispersion decay and scattering theory. John Willey \& Sons, Hoboken, New Jersey, 2012.
\vspace{-3 mm}
\bibitem{KK2020} 
A. Komech, E. Kopylova, 
Attractors of nonlinear Hamiltonian partial differential equations, 
{\em Russ. Math. Surv.} {\bf 75} (2020), no. 1, 1--87.
\vspace{-3 mm}
\bibitem{KK2022}
A. Komech, E. Kopylova, 
Attractors of Hamiltonian nonlinear partial differential equations, Cambridge University Press, Cambridge, 2022.
\vspace{-3 mm}
\bibitem{KKS2011}
A. Komech, E. Kopylova, H. Spohn, 
Scattering of solitons for Dirac equation coupled to a particle, 
{\em J. Math. Analysis and Appl.} {\bf 383} (2011), no. 2, 265--290
\vspace{-3 mm}
\bibitem{KKS2018}
A. Komech, E. Kopylova, H. Spohn, 
On global attractors and radiation damping for nonrelativistic particle coupled to scalar field. 
{\em St. Petersburg Math. J.} {\bf 29} (2018), no. 2, 249--266.
\vspace{-8mm}
\bibitem{KK2023}
E. Kopylova,  A. Komech,
On asymptotic stability of solitons for 2D Maxwell--Lorentz equations,  {\em J. Math. Phys.} {\bf 64}  (2023), 101504.
\vspace{-3 mm}
\bibitem{KK2023+}
E. Kopylova,  A. Komech,
Global attraction to solitons for 2D Maxwell--Lorentz equations with spinning particle,
{\em St. Petersburg Math. J.} {\bf 35} (2024), 827--838. 
\vspace{-3 mm}
\bibitem{KK2024}
A. Komech, E. Kopylova, 
On orbital stability of solitons for 2D Maxwell--Lorentz equations. 
{\em Comm. Pure and  Appl. Anal.}  {\bf 24} (2024), 325--338.
\vspace{-3 mm}
\bibitem{K2010}
E. Kopylova,
Dispersive estimates for the 2D wave equation, {\em Russian J. Math. Phys.} {\bf 17} (2010), no. 2, 226-239.
\vspace{-3 mm}
\bibitem{O}
F. W. J. Olver \& al., NIST Handbook of Mathematical Functions, Cambridge University Press, Cambridge, 2010.
\vspace{-3 mm}
\bibitem{S2005} 
W Schlag, Dispersive estimates for Schr\"odinger operators in dimension two. 
{\em Comm. Math. Phys.} {\bf 257} (2005), no. 1, 87--117.
\vspace{-3 mm}
\bibitem{S2004}
 H. Spohn, Dynamics of Charged Particles and Their Radiation Field, Cambridge
University Press, Cambridge, 2004.

\end{thebibliography}
\end{document}